\documentclass[12pt]{article}
\usepackage{amssymb}
\usepackage{amsmath}
\usepackage{amsthm}
\usepackage{epsfig}
\usepackage{mathrsfs}
\usepackage{graphicx}
\usepackage{ textcomp }
\usepackage{xcolor}
\usepackage{hyperref}
\usepackage{bbold}
\usepackage{graphicx}
\usepackage{caption}
\usepackage{subcaption}
\usepackage{authblk}
\usepackage{fullpage}

\def\R{\mathbb{R}}

\def\la{\langle}
\def\ra{\rangle}

\def\to{\rightarrow}

\def\eps{\varepsilon}

\newtheorem{theorem}{Theorem}[section]
\newtheorem{lemma}[theorem]{Lemma}
\newtheorem{proposition}[theorem]{Proposition}
\newtheorem{definition}[theorem]{Definition}

\newtheorem{remark}[theorem]{Remark}
\newtheorem{rk&ex}[theorem]{Remarks \& Examples}
\newtheorem{corollary}[theorem]{Corollary}

\newcommand{\beqar}{\begin{eqnarray*}}
\newcommand{\eeqar}{\end{eqnarray*}}
\newcommand{\beqarl}{\begin{eqnarray}}
\newcommand{\eeqarl}{\end{eqnarray}}
\newcommand{\lp}{\left(}
\newcommand{\rp}{\right)}

\newcommand{\be}{\begin{equation}}
\newcommand{\ee}{\end{equation}}

\newcommand{\fe}{f_{\varepsilon}}

\newcommand{\nn}{\nonumber}
\newcommand{\asymn}{{[\mathbf{n}]_{\times}}}
\newcommand{\nvec}{\mathbf{n}}
\newcommand{\rvec}{\mathbf{r}}
\newcommand{\vezero}{\mathbf{e_{1}}}
\newcommand{\uu}{\mathbf{u}}
\newcommand{\vv}{\mathbf{v}}
\newcommand{\ww}{\mathbf{w}}
\newcommand{\Id}{\mathrm{Id}}
\newcommand{\PD}{\mbox{PD}}
\newcommand{\tr}{\mathrm{tr}}

\title{A new flocking model through body attitude coordination}

\author[(1)]{Pierre Degond}
\author[(2)]{Amic Frouvelle}
\author[(3)]{Sara Merino-Aceituno}

\affil[(1)(3)]{Department of Mathematics, Imperial College London, 
London SW7 2AZ, UK}
\affil[(1)]{pdegond@imperial.ac.uk}
\affil[(3)]{s.merino-aceituno@imperial.ac.uk}

\affil[(2)]{CEREMADE, UMR CNRS 7534, Universit\'e de Paris-Dauphine, PSL Research University\\

Place du Mar\'echal De Lattre De Tassigny, 75775 PARIS CEDEX 16 - FRANCE\\

frouvelle@ceremade.dauphine.fr
}
\begin{document}

\maketitle
\begin{abstract}
We present a new model for multi-agent dynamics where each agent is described by its position and body attitude: agents travel at a constant speed in a given direction and their body can rotate around it adopting different configurations. In this manner, the body attitude is described by three orthonormal axes giving an element in~$SO(3)$ (rotation matrix). Agents try to coordinate their body attitudes with the ones of their neighbours. In the present paper, we give the Individual Based Model (particle model) for this dynamics and derive its corresponding kinetic and macroscopic equations.

The work presented here is inspired by the Vicsek model and its study in \cite{degond2008continuum}. This is a new model where collective motion is reached through body attitude coordination.

\bigskip



\end{abstract}

\tableofcontents

\section{Introduction}

In this paper we model collective motion where individuals or agents are described by their position and body attitude. The body attitude is given by three orthonormal axis; one of the axes describes the direction in which the agent moves at a constant speed; the other two axis indicate the relative position of the body with respect to this direction. Agents try to coordinate their body attitude with those of near neighbours (see Figure~\ref{fig:coordination}). Here we present an Individual Based Model (particle model) for body attitude coordination and derive the corresponding macroscopic equations from the associated mean-field equation, which we refer to as the Self-Organized Hydrodynamics  for body attitude coordination (SOHB), by reference to the Self-Organized Hydrodynamics (SOH) derived from the Vicsek dynamics (see \cite{degond2008continuum} and discussion below).

There exists already a variety of models for collective behaviour depending on the type of interaction between agents. However, to the best of our knowledge, this is the first model that takes into account interactions based on body attitude coordination. This has applications in the study of collective motion of animals such as birds and fish and it is a stepping stone to model more complex agents formed by articulated bodies (corp\-ora)~\cite{constantin2010onsager,constantin2010high}. In this section we present related results in the literature and the structure of the document. 

\bigskip

\renewcommand\thempfootnote{\arabic{mpfootnote}} 
\begin{figure}
  \begin{minipage}{\textwidth}
    \centering
    \begin{subfigure}[b]{0.4\textwidth}
        \includegraphics[trim={0 4cm 0 0},clip,width=\textwidth]{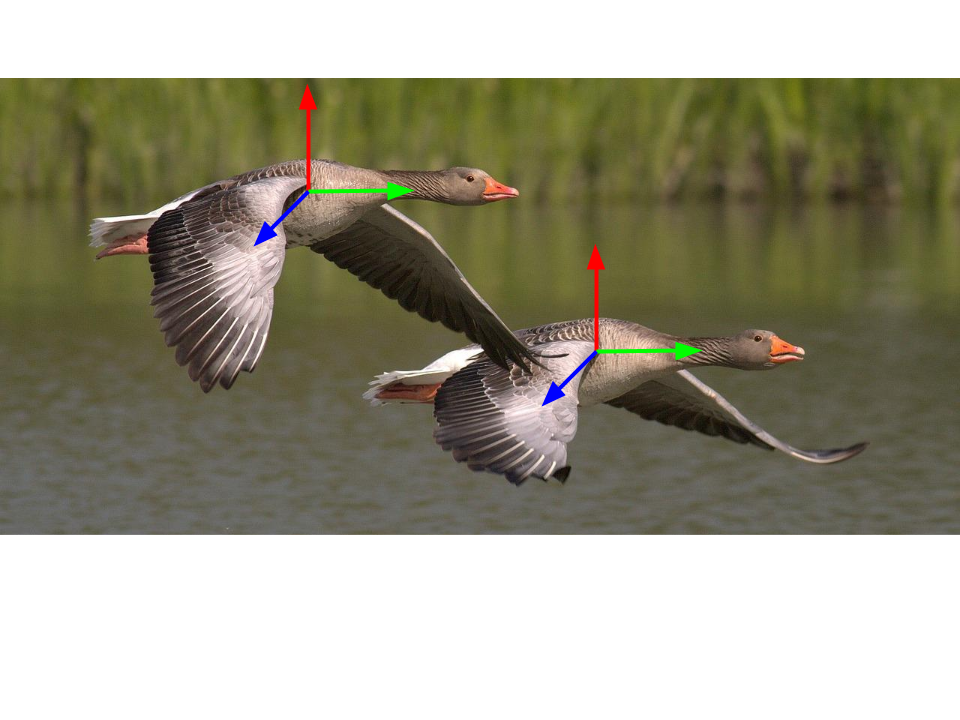}
        \caption{Birds with coordinated body attitude. Three orthonormal axis describe the body attitude: the green arrow indicates the direction of movement; the blue and red ones indicate the position of the body with respect to this direction.}
    \end{subfigure}
    ~ 
    \begin{subfigure}[b]{0.4\textwidth}
        \includegraphics[trim={4cm 5cm 4cm 0},clip,width=\textwidth]{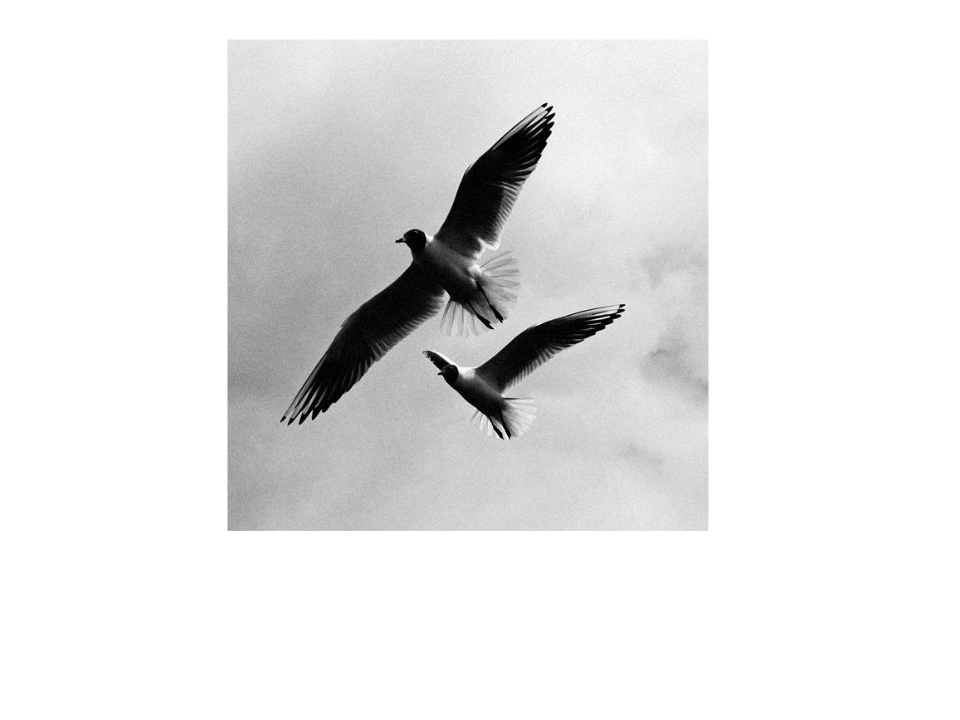}
        \caption{Birds with  no coordinated body attitude.}
    \end{subfigure}
    ~ 

    \begin{subfigure}[b]{0.6\textwidth}
        \includegraphics[trim={0 4cm 0 0},clip,width=\textwidth]{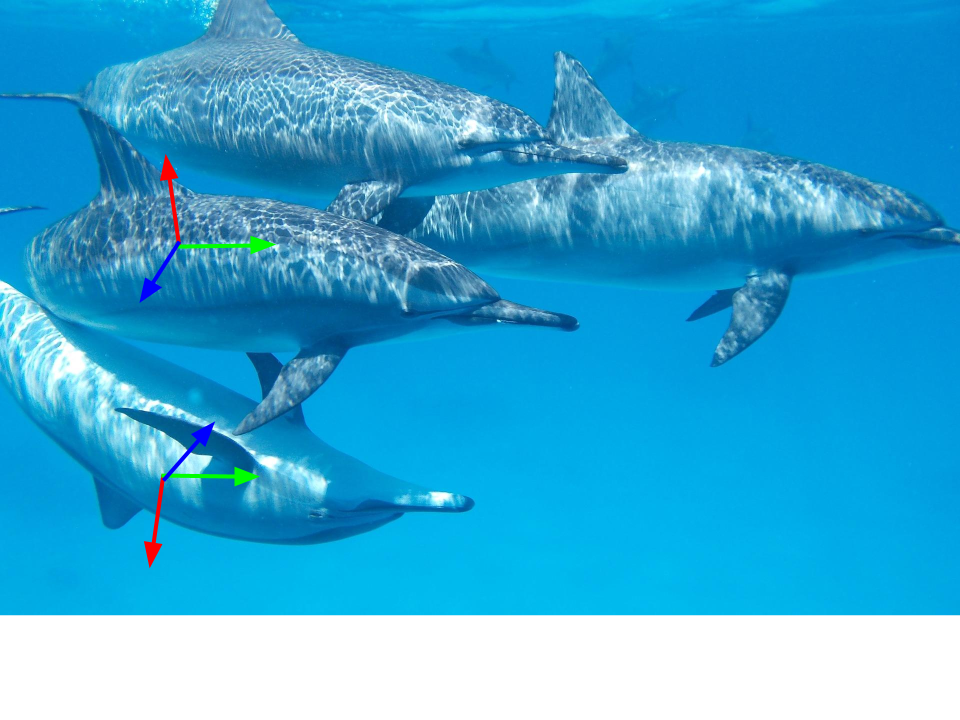}
        \caption{Dolphins moving in the same direction but with different body attitude. In this example one can see that the body attitude coordination model gives more information than the Vicsek model (which only describes the direction of movement).}
    \end{subfigure}
    \caption[Examples]{Examples of body attitude coordination/dis-coordination and the use of the rotation matrix representation.\footnote{These images are in public domain (released under Creative Commons CC0 by pixabay.com).}}\label{fig:coordination}
  \end{minipage}
\end{figure}

There exists a vast literature on collective behaviour. In particular, here we deal with the case of self-propelled particles which is ubiquitous in nature. It includes, among others, fish schools, flocks of birds, herds \cite{buhl2006disorder,cavagna2010scale,parrish1997animal}; bacteria \cite{ben2000cooperative,zhang2010collective}; human walking behaviour~\cite{helbing2007dynamics}.  
The interest in studying these systems is to gain understanding on the emergent properties: local interactions give rise to large scale structures in the form of patterns and phase transitions (see the review \cite{vicsek2012collective}). These large scale structures take place when the number of individuals or agents is very large. In this case a statistical description of the system is more pertinent than an agent-based one. With this in mind mean-field limits are devised when the number of agents tend to infinity. From  them macroscopic equations can be obtained using hydrodynamic limit techniques (as we explain below).

\medskip
The results presented here are inspired from those in reference \cite{degond2008continuum}. There the authors consider the Vicsek model which is a particular type of model for self-propelled particles \cite{aldana2003phase,couzin2002collective,gregoire2004onset,vicsek1995novel}. The Vicsek model describes collective motion where agents travel at a constant speed in a given direction. At each time step the direction of movement is updated to the averaged one of the neighbouring agents, with some noise. The position is updated considering the distance travelled during that time step.

 Our results here are inspired by the Self-Organized Hydrodynamics (SOH)  model (the continuum version of the Vicsek model), where we have substituted velocity alignment by body attitude coordination. Other refinements and adaptations of the Vicsek model (at the particle level) or the SOH model (at the continuum level) have been proposed in the literature, we just mention the following ones as examples: in \cite{cavagna2014flocking} an individual-based model is proposed to better describe collective motion of turning birds; in \cite{degond2015multi} agents are considered to have the shape of discs and volume exclusion is included in the dynamics; a description of nematic alignment in rods is done in \cite{degond2015continuum}.

In \cite{degond2008continuum} the authors  investigate the mean-field limit and macroscopic limit of the Vicsek model. The mean-field limit gives a kinetic equation that takes the form of a Fokker-Planck equation referred to as the mean-field limit Vicsek model. 

To obtain the macroscopic equations (the SOH model), the authors in \cite{degond2008continuum} use the well-known tools of hydrodynamic limits, first developed in the framework of the Boltzmann equation for rarefied gases \cite{cercignani2013mathematical,degond2004macroscopic,sone2012kinetic}. Since its first appearance, hydrodynamics limits have been used in other different contexts, including traffic flow modeling \cite{aw2002derivation,helbing2001traffic} and supply chain research~\cite{armbruster2006model,degond2007stochastic}. However, in \cite{degond2008continuum} a new concept is introduced: the Generalized Collision Invariant (GCI). Typically to obtain the macroscopic equations we require as many conserved quantities in the kinetic equation as the dimension of the equilibria (see again \cite{vicsek2012collective}). In the mean-field limit Vicsek model this requirement is not fulfilled and the GCI is used to sort out this problem. For other cases where the GCI concept has been used see \cite{degond2014hydrodynamics,degond2014macroscopic,degond2012hydrodynamics,degond2014evolution,degond2015self,frouvelle2012continuum}.

\bigskip
After this introduction and the following discussion about the main result, the next part of the paper is dedicated to the modelling. In Section~\ref{sec:derivation_IBM} we explain the derivation of the Individual Based Model for body coordination dynamics:  given~$N$ agents labelled by~$k=1,\hdots, N$ the positions and body attitudes~$(\mathbf{X_k}, A_k) \in \R^3\times SO(3)$ over time are given by the Stochastic Differential Equations \eqref{eq:IBM}-\eqref{eq:IBM2}. In Section~\ref{sec:mean_field_limit} we give the corresponding (formal) mean-field limit (Prop.~\ref{prop:mean_field_limit}) for the evolution of the empirical measure when the number of agents~$N\to \infty$. 

The last part concerns the derivation of the macroscopic equations (Theorem~\ref{th:macro_limit}) for the total density of the particles~$\rho=\rho(t,x)$ and the matrix of the mean body attitude~$\Lambda = \Lambda(t,x)$. To obtain these equations we first study the rescaled mean-field equation (Eq. \eqref{eq:MFlimit} in Section~\ref{sec:scaling}), which is, at leading order, a Fokker-Planck equation. We determine its equilibria, which are given by a von Mises distribution on~$SO(3)$ (Eq. \eqref{eq:Von_Mises_equilibria}, Section~\ref{sec:FP}). Finally in Section~\ref{sec:GCI} we obtain the Generalized Collision Invariants (Prop.~\ref{prop:non_constant_GCI}), which are the main tool to be able to derive the macroscopic equations in Section~\ref{sec:macro-limit}.

\section{Discussion of the main result: the Self-Organized Hydrodynamics for body attitude coordination (SOHB)}
\label{sec:discussion_macro}

The main result of this paper is Theorem~\ref{th:macro_limit} which gives the following macroscopic equations for  the density of agents~$\rho=\rho(t,x)$ and  the matrix of the mean body attitude~$\Lambda = \Lambda(t,x)\in SO(3)$ (i.e., the Self-Organized Hydrodynamics for body attitude coordination (SOHB)):
\begin{align} \label{eq:macro1}
&\partial_t \rho + c_1 \nabla_x \cdot ( \rho\, \Lambda \vezero) =0,\\
&\rho \Big( \partial_t\Lambda+ c_2 \big((\Lambda \vezero) \cdot \nabla_x\big)\Lambda \Big) +\left[ (\Lambda \vezero) \times \big(c_3\nabla_x \rho+c_4\rho\,\rvec_x(\Lambda)\big) + c_4\rho\,\delta_x(\Lambda)\Lambda \vezero\right]_\times \Lambda =0. \label{eq:macro2} 
\end{align} 
In the equations above~$c_1$,~$c_2$,~$c_3$ and~$c_4$ are explicit constants which depend on the parameters of the model (namely the rate of coordination and the level of noise). The expressions of the constants~$c_2$,~$c_3$ and~$c_4$ depend on the Generalized Collision Invariant mentioned in the introduction (and computed thanks to Prop.~\ref{prop:non_constant_GCI}). The constant~$c_1$ is obtained as a ``consistency'' relation (Lemma~\ref{lem:consistency_relation}). 
In \eqref{eq:macro2}, the operation~$[\cdot]_\times$ transforms a vector~$\vv$ in an antisymmetric matrix such that~$[\vv]_\times \uu=\vv\times \uu$ for any vector~$\uu$ (see \eqref{eq:def_operator_asym} for the exact definition). The scalar~$\delta_x(\Lambda)$ and the vector~$\rvec_x(\Lambda)$ are first order differential operators intrinsic to the dynamics : if~$\Lambda(x)=\exp\lp[\mathbf{b}(x)]_\times \rp\Lambda(x_0)$ with~$\mathbf{b}$ smooth around~$x_0$ and~$\mathbf{b}(x_0)=0$, then 
\[
\delta_x(\Lambda)(x_0) = \left.\nabla_x \cdot \mathbf{b}(x)\right|_{x=x_0}\quad\text{ and }\quad
\rvec_x(\Lambda)(x_0) = \left.\nabla_x \times \mathbf{b}(x)\right|_{x=x_0},
\]
where~$\nabla_x\times$ is the curl operator. These operators are well-defined as long as~$\Lambda$ is smooth: as we will see in the next section, we can always express a rotation matrix as~$\exp\lp[\mathbf{b}]_\times \rp$ for some vector~$\mathbf{b}\in \R^3$, and this function~$\mathbf{b}\mapsto\exp\lp[\mathbf{b}]_\times \rp$ is a local diffeomorphism between a neighborhood of~$0\in\R^3$ and the identity of~$SO(3)$. This gives a unique smooth representation of~$\mathbf{b}$ in the neighborhood of~$0$ when~$x$ is in the neighborhood of~$x_0$ since then~$\Lambda(x)\Lambda(x_0)^{-1}$ is in the neighborhood of~$\Id$.

\medskip

We express Eq. \eqref{eq:macro2} in terms of the basis vectors~$\{\Omega=\Lambda \mathbf{e_1},\uu=\Lambda \mathbf{e_2},\vv=\Lambda \mathbf{e_3}\}$ and we write~$\Lambda=\Lambda(\Omega,\uu,\vv)$. System~\eqref{eq:macro1}-\eqref{eq:macro2} can be expressed as an evolution system for~$\rho$ and the basis~$\{\Omega,\uu,\vv\}$ as follows:
\begin{align}
&\partial_t \rho + c_1 \nabla_x \cdot ( \rho\, \Omega) =0,\label{eq:rho_aux}\\
&\rho D_t\Omega + P_{\Omega^\perp}\lp c_3 \nabla_x\rho + c_4 \rho\, \rvec \rp =0, \label{eq:Omega_aux}\\
& \rho D_t \uu - \uu \cdot \lp c_3 \nabla_x \rho + c_4 \rho \,\rvec\rp \Omega + c_4\rho \,\delta\, \vv =0,\label{eq:u_aux}\\
&\rho D_t \vv - \vv \cdot \lp c_3 \nabla_x \rho + c_4 \rho\, \rvec\rp \Omega - c_4 \rho \,\delta\, \uu =0, \label{eq:v_aux}
\end{align}
where~$D_t := \partial_t + c_2 (\Omega\cdot \nabla_x)$,~$\delta=\delta_x(\Lambda(\Omega,\uu,\vv))$  and~$\rvec=\rvec_x(\Lambda(\Omega,\uu,\vv))$. The operator~$P_{\Omega^\perp}$ denotes the projection on the orthogonal of~$\Omega$. We easily see here that these equations preserve the constraints~$|\Omega|=|\uu|=|\vv|=1$ and~$\Omega\cdot\uu=\Omega\cdot\vv=\uu\cdot\vv=0$.
The expressions of~$\delta$ and~$\rvec$ are:
\beqar
\delta &=& [(\Omega \cdot \nabla_x) \uu] \cdot \vv + [(\uu\cdot \nabla_x)\vv]\cdot \Omega + [(\vv\cdot \nabla_x)\Omega]\cdot \uu, \\
\rvec &=&  (\nabla_x\cdot \Omega) \Omega + (\nabla_x \cdot \uu) \uu + (\nabla_x\cdot \vv) \vv .  
\eeqar

\bigskip
Eq. \eqref{eq:macro1} (or equivalently Eq. \eqref{eq:rho_aux}) is the continuity equation for~$\rho$ and ensures mass conservation. The convection velocity is given by~$c_1\Lambda \vezero=c_1\Omega$ and $\Omega$ gives the direction of motion.
Eq. \eqref{eq:macro2} (or equivalently Eqs. \eqref{eq:Omega_aux}-\eqref{eq:v_aux}) gives the evolution of~$\Lambda$. 
We remark that every term in  Eq. \eqref{eq:macro2} belongs to the tangent space at~$\Lambda$ in~$SO(3)$; this is true for the first term since~$(\partial_t + c_2 (\Lambda \vezero) \cdot \nabla_x)$ is a differential operator and it also holds for the second term because it is the product of an antisymmetric matrix with~$\Lambda$ (see Prop.~\ref{prop:tangentspaceSOn}). Alternately, this means that~$(\Omega(t),\uu(t),\vv(t))$ is a direct orthonormal basis as soon as~$(\Omega(0),\uu(0),\vv(0))$.

\medskip
The term corresponding to~$c_3$ in \eqref{eq:macro2} gives the influence of~$\nabla_x \rho$ (pressure gradient) on the body attitude~$\Lambda$. It has the effect of rotating the body around the vector directed by~$(\Lambda\vezero) \times  \nabla_x \rho$ at an angular speed given by~$\frac{c_3}\rho\|(\Lambda\vezero) \times  \nabla_x \rho\|$, so as to align~$\Omega$ with~$-\nabla_x\rho$. Indeed the solution of the differential equation~$\frac{d\Lambda}{dt}+\gamma[\nvec]_\times \Lambda=0$, when~$\nvec$ is a constant unit vector and~$\gamma$ a constant scalar, is given by~$\Lambda(t)=\exp(-\gamma\, t[\nvec]_\times)\Lambda_0$, and~$\exp(-\gamma\, t[\nvec]_\times)$ is the rotation of axis~$\nvec$ and angle~$-\gamma\, t$ (see~\eqref{eq:exponential_form}, it is called Rodrigues' formula).
Since we will see that~$c_3$ is positive the influence of this term consists of relaxing the direction of displacement~$\Lambda \vezero$ towards~$\nabla_x\rho$. 
Alternately, we can see from~\eqref{eq:Omega_aux} that~$\Omega$ turns in the opposite direction to~$\nabla_x\rho$, showing that the~$\nabla_x\rho$ term has the same effect as a pressure gradient in classical hydrodynamics. We note that the pressure gradient has also the effect of rotating the whole body frame (see influence of~$\nabla_x\rho$ on~$\uu$ and~$\vv$) just to keep the frame orthonormal.
Similarly to what happens with the~$\nabla_x\rho$ term in Eq. \eqref{eq:macro2}, the term containing~$c_4\rho\,\rvec$ in Eq. \eqref{eq:Omega_aux} has the effect of relaxing the direction of displacement~$\Omega$ towards~$-\rvec$ (we will indeed see that~$c_4$ is positive).
Finally, the last terms of Eqs. \eqref{eq:u_aux}-\eqref{eq:v_aux} have the effect of rotating the vectors~$\uu$ and~$\vv$ around~$\Omega$ along the flow driven by~$D_t$ at angular speed~$c_4\delta$. 

If we forget the term proportional to~$\rvec$ in~\eqref{eq:Omega_aux}, System~\eqref{eq:rho_aux}-\eqref{eq:Omega_aux} is decoupled from~\eqref{eq:u_aux}-\eqref{eq:v_aux}, and is an autonomous system for~$\rho$ and~$\Omega$, which coincides with the Self-Organized Hydrodynamic (SOH) model. The SOH model provides the fluid description of a particle system obeying the Vicsek dynamics~\cite{degond2008continuum}. As already discussed in~\cite{degond2008continuum}, the SOH model bears analogies with the compressible Euler equations, where~\eqref{eq:rho_aux} is obviously the mass conservation equation and~\eqref{eq:Omega_aux} is akin to the momentum conservation equation, where momentum transport~$\rho D_t\Omega$ is balanced by a pressure force~$-P_{\Omega^\perp}\nabla_x\rho$. There are however major differences. The first one is the presence of the projection operation~$P_{\Omega^\perp}$ which is there to preserve the constraint~$|\Omega|=1$. Indeed, while the velocity in the Euler equations is an arbitrary vector, the quantity~$\Omega$ in the SOH model is a velocity orientation and is normalized to~$1$. The second one is that the convection speed~$c_2$ in the convection operator~$D_t$ is a priori different from the mass convection speed~$c_1$ appearing in the continuity equation. This difference is a signature of the lack of Galilean invariance of the system, which is a common feature of all dry active matter models.

The major novelty of the present model, which can be referred to as the Self-Organized Hydrodynamic model with Body coordination (or SOHB) is that the transport of the direction of motion~$\Omega$ involves the influence of another quantity specific to the body orientation dynamics, namely the vector~$\rvec$. The overall dynamics tends to align the velocity orientation~$\Omega$, not opposite to the density gradient~$\nabla_x\rho$ but opposite to a composite vector~$(c_3\nabla_x\rho+c_4\rho\,\rvec)$. The vector~$\rvec$ is the rotational of a vector~$\mathbf{b}$ locally attached to the frame (namely the unit vector of the local rotation axis multiplied by the local angle of rotation around this axis). This vector gives rise to an effective pressure force which adds up to the usual pressure gradient. It would be interesting to design specific solutions where this effective pressure force has a demonstrable effect on the velocity orientation dynamics.

In addition to this effective force, spatial inhomogeneities of the body attitude also have the effect of inducing a proper rotation of the frame about the direction of motion. This proper rotation is also driven by spatial inhomogeneities of the vector~$\mathbf{b}$ introduced above, but are now proportional to its divergence.

\section{Modelling: the Individual Based Model and its mean-field limit}
\label{sec:modelling}

The body attitude is given by a rotation matrix. Therefore, we work on the Riemannian manifold~$SO(3)$ (Special Orthogonal Group), which is formed by the subset of matrices~$A$ such that~$AA^T=\Id$ and~$\det(A)=1$, where~$\Id$ stands for the identity matrix. 

\medskip

In this document
$\mathcal{M}$ indicates the set of square matrices of dimension~$3$; 
$\mathcal{A}$ is the set of antisymmetric matrices of dimension~$3$;
$\mathcal{S}$ is the set of symmetric matrices of dimension~$3$. Typically we will denote by~$A, \Lambda$ matrices in~$SO(3)$ and by~$P$ matrices in~$\mathcal{A}$. Bold symbols~$\mathbf{n}, \vv,\vezero$ indicate vectors.

\medskip

We will often use the so-called Euler axis-angle parameters to represent an element in~$SO(3)$: to~$A\in SO(3)$ there is associated an angle~$\theta \in [0,\pi]$ and a vector~$\nvec\in S^2$ so that~$A=A(\theta,\nvec)$ corresponds to the anticlockwise rotation of angle~$\theta$ around the vector~$\nvec$. It is easy to see that 
\be \label{eq:trace_theta}
\tr(A)=1+2\cos\theta
\ee
 (for instance expressing~$A$ in an orthonormal basis with~$\nvec$), so the angle~$\theta$ is uniquely defined as~$\arccos(\frac12(\tr(A)-1))$. Notice that~$\nvec$ is uniquely defined whenever~$\theta\in(0,\pi)$ (if~$\theta=0$ then~$\nvec$ can be any vector in~$S^2$ and if~$\theta = \pi$ then the direction of~$\nvec$ is uniquely defined but not its orientation). For a given vector~$\uu$, we introduce the antisymmetric matrix~$\left[ \uu\right]_\times$, where~$[\cdot]_\times$ is the linear operator from~$\mathbb{R}^3$ to~$\mathcal{A}$ given by 
\be \label{eq:def_operator_asym}
\left[ \uu\right]_\times := \left[
\begin{array}{ccc}
0 & -u_3 & u_2\\
u_3 & 0 & -u_1\\
-u_2 & u_1 & 0
\end{array}
\right],
\ee  
 so that for any vectors~$\uu,\vv\in \R^3$, we have~$\left[ \uu\right]_\times \vv= \uu\times \vv$. In this framework, we have the following representation for~$A\in SO(3)$ (called Rodrigues' formula): 
\beqarl 
A=A(\theta,\nvec)&=& \Id + \sin\theta\, \asymn + (1-\cos\theta) [\nvec]_\times^2 \label{eq:Rodrigues_formula} \\
&=& \exp(\theta \asymn) \label{eq:exponential_form}.
\eeqarl
 We also have~$\nvec\times(\nvec\times\vv)=(\nvec\cdot\vv)\,\nvec-(\nvec\cdot\nvec)\vv$, therefore when~$\nvec$ is a unit vector, we have :
\begin{equation}\label{eq-asym2-proj}
[\nvec]_\times^2=\nvec\otimes\nvec-\Id,
\end{equation}
where the tensor product~$\mathbf{a}\otimes\mathbf{b}$ is the matrix defined by~$(\mathbf{a}\otimes\mathbf{b})\uu=(\uu\cdot\mathbf{b}) \mathbf{a}$ for any~$\uu\in\R^3$.
Finally,~$SO(3)$ has a natural Riemannian metric (see~\cite{huynh2009metrics}) induced by the following inner product in the set of square matrices of dimension~$3$:
\be \label{eq:dot_product_SO3} 
A \cdot B=\frac{1}{2}\,\tr(A^T B)=\frac{1}{2} \sum_{i,j}A_{ij}B_{ij}.
\ee
This normalization gives us that for any vectors~$\uu,\vv\in \R^3$, we have that
\begin{equation}\label{eq:properties_asym_matrix}
\left[ \uu\right]_\times \cdot \left[ \vv\right]_\times = (\uu \cdot \vv).
\end{equation}
Moreover, the geodesic distance on~$SO(3)$ between~$\Id$ and a rotation of angle~$\theta\in[0,\pi]$ is exactly given by~$\theta$ (the geodesic between~$\Id$ and~$A$ is exactly~$t\in[0,\theta]\mapsto\exp(t\asymn)$). See Appendix~\ref{ap:SO(3)} for some properties of~$SO(3)$ used throughout this work.

Seeing~$SO(3)$ as a Riemannian manifold, we will use the following notations:~$T_A$ is the tangent space in~$SO(3)$ at~$A\in SO(3)$;~$P_{T_A}$ denotes the orthogonal projection onto~$T_A$; the operators~$\nabla_A , \nabla_A \cdot$ are the gradient and divergence in~$SO(3)$, respectively. These operators are computed in Section~\ref{sec:differential_calculus} in the Euler axis-angle coordinates.

\subsection{The Individual Based Model}
\label{sec:derivation_IBM}

Consider~$N$ agents labelled by~$k=1, \hdots, N$ with positions~$\mathbf{X_k}(t) \in \mathbb{R}^3$ and associated matrices (body attitudes)~$A_k(t) \in SO(3)$. For each~$k$, the three unit vectors representing the frame correspond to the vectors of the matrix~$A_k(t)$ when written as a matrix in the canonical basis~$(\vezero, \mathbf{e_{2}}, \mathbf{e_{3}})$ of~$\mathbb{R}^3$. In particular, the direction of displacement of the agent is given by its first vector~$A_k(t)\vezero$.
\bigskip

	 \textit{Evolution of the positions.} Agents move in the direction of the first axis with constant speed~$v_0$
$$\frac{d\mathbf{X_k}(t)}{dt}=v_0 A_k(t) \vezero.$$

\bigskip
 \textit{Evolution of the body attitude matrix.} Agents try to coordinate their body attitude with those of their neighbours. So we are facing two different problems from a modelling viewpoint, namely to define the target body attitude, and to define the way agents relax their own attitude towards this ``average'' attitude.

As for the Vicsek model~\cite{degond2008continuum}, we consider a kernel of influence~$K=K(x)\geq 0$ and define the matrix 
\be \label{eq:Mk}
   \mathbb{M}_k(t) := \frac{1}{N}\sum_{i=1}^N K(|\mathbf{X}_i(t)-\mathbf{X_k}(t)|) A_i(t).
\ee
This matrix corresponds to the averaged body attitude of the agents inside the zone of influence corresponding to agent~$k$. Now~$\mathbb{M}_k(t) \notin SO(3)$, so we need to orthogonalize and remove the dilations, in order to construct a target attitude in~$SO(3)$. We will see that the polar decomposition of~$\mathbb{M}_k(t)$ is a good choice in the sense that it minimizes a weighted sum of the squared distances to the attitudes of the neighbours. We also refer to~\cite{moakher2002means} for some complements on averaging in~$SO(3)$.

\medskip

We give next the definition of polar decomposition:
\begin{lemma}[Polar decomposition of a square matrix, {\cite[Section 4.2.10]{golub2012matrix}}]~
\label{lem:polar_decomposition}

Given a matrix~$M\in \mathcal{M}$, if~$\det(M)\neq 0$ then there exists a unique orthogonal matrix~$A$ (given by~$A= M(\sqrt{M^TM})^{-1}$) and a unique symmetric positive definite matrix~$S$ such that~$M=AS$.
\end{lemma}

\begin{proposition} Suppose that the matrix~$\mathbb{M}_k(t)$ has positive determinant. Then the following assertions are equivalent:
\label{prop:polar-decomp-as-minimization}
\begin{enumerate}
\item[(i)] The matrix~$A$ minimizes the quantity~$\frac{1}{N}\sum_{i=1}^N K(|\mathbf{X}_i(t)-\mathbf{X_k}(t)|) \|A_i(t)-A\|^2$ among the elements of~$SO(3)$.
\item[(ii)] The matrix~$A$ is the element of~$SO(3)$ which maximizes the quantity~$A\cdot \mathbb{M}_k(t)$.
\item[(iii)] The matrix~$A$ is the polar decomposition of~$\mathbb{M}_k(t)$.
\end{enumerate}
\end{proposition}
\begin{proof}
We get the equivalence between the first two assertions by expanding:
\begin{equation*}
\|A_i(t)-A\|^2=\frac12[\tr(A_i(t)^TA_i(t))+\tr(A^TA)]-2A\cdot A_i(t)=3-2A\cdot A_i(t),
\end{equation*}
since~$A$ and~$A_i(t)$ are both orthogonal matrices. So minimizing the weighted sum of the squares distances amounts to maximizing inner product of~$A$ and the weighted sum~$\mathbb{M}_k$ of the matrices~$A_i$ given by~\eqref{eq:Mk}.

Therefore if~$\det \mathbb{M}_k>0$, and~$A$ is the polar decomposition of~$\mathbb{M}_k$, we immediately get that~$\det A>0$, hence~$A\in SO(3)$. We know that~$S$ can be diagonalized in an orthogonal basis :~$S=P^TDP$ with~$P^TP=\Id$ and~$D$ is a diagonal matrix with positive diagonal elements~$\lambda_1,\lambda_2,\lambda_3$.
Now if~$B\in SO(3)$ maximizes~$\frac12\tr(B^T\mathbb{M}_k)$ among all matrices in~$SO(3)$, then it maximizes~$\tr(B^TAP^TDP)=\tr(PB^TAP^TD)$. So the matrix~$\bar{B}=PB^TAP^T$ maximizes~$\tr(\bar{B}D)=\lambda_1\bar{b}_{11}+\lambda_2\bar{b}_{22}+\lambda_3\bar{b}_{33}$ among the elements of~$SO(3)$ (the map~$B\mapsto PB^TAP^T$ is a one-to-one correspondence between~$SO(3)$ and itself). But since~$\bar{B}$ is an orthogonal matrix, all its column vectors are unit vectors, and so~$b_{ii}\leqslant 1$, with equality for~$i=1,2$ and~$3$ if and only if~$\bar{B}=\Id$, that is to say~$PB^TAP^T=\Id$, which is exactly~$B=A$.
\end{proof}

We denote by~$\PD(\mathbb{M}_k)\in O(3)$ the corresponding orthogonal matrix coming from the Polar Decomposition of~$\mathbb{M}_k$.
	 
We now have two choices for the evolution of~$A_k$. We can use the second point of Proposition~\ref{prop:polar-decomp-as-minimization} and follow the gradient of the function to maximize :
\begin{equation}
\frac{dA_k(t)}{dt}= \left.\nu \nabla_{A}(\mathbb{M}_k\cdot A)\right|_{A=A_k} = \nu P_{T_{A_k}}\mathbb{M}_k,\label{eq:nopolar}
\end{equation}
 (see~\eqref{eq:gradient_product} for the last computation,~$P_{T_{A_k}}$ is the projection on the tangent space, this way the solution of the equation stays in~$SO(3)$). 

Or we can directly relax to the polar decomposition~$\PD(\mathbb{M}_k)$, in the same manner:

$$ \frac{dA_k(t)}{dt}= \nu P_{T_{A_k}} \lp \PD(\mathbb{M}_k) \rp.$$

We can actually see that the trajectory of this last equation, when~$\PD(\mathbb{M}_k)$ belongs to~$SO(3)$ and does not depend on~$t$, is exactly following a geodesic (see Prop.~\ref{prop:relax_geodesic}). Therefore in this paper we will focus on this type of coordination. The positive coefficient~$\nu$ gives the intensity of coordination, in the following we will assume that it is a function of the distance between~$A_k$ and~$\PD(\mathbb{M}_k)$ (the angle of the rotation~$A_k^T\PD(\mathbb{M}_k)$), which is equivalent to say that~$\nu$ depends on~$A_k\cdot\PD(\mathbb{M}_k)$.

	\begin{remark} Some comments:
	\label{rem:IBM}
	\begin{enumerate}
		\item One could have used the Gram-Schmidt orthogonalization instead of the Polar Decomposition, but it depends on the order in which the vector basis is taken (for instance if we start with~$\vezero$, it would define the first vector as the average of all the directions of displacement, independently of how the other vectors of the body attitudes of the individuals are distributed). The Polar Decomposition gives a more canonical way to do this. 
		\item We expect that the orthogonal matrix coming from the Polar Decomposition of~$\mathbb{M}_k$ belongs in fact to~$SO(3)$. Firstly, we notice that~$O(3)$ is formed by two disconnected components:~$SO(3)$ and the other component formed by the matrices with determinant -1. We assume that the motion of the agents is smooth enough so that the average~$\mathbb{M}_k$ stays `close' to~$SO(3)$ and that, in particular,~$\det(\mathbb{M}_k)> 0$. 

A simple example is when we only average two different matrices~$A_1$ and~$A_2$ of~$SO(3)$. We then have~$\mathbb{M}=\frac12(A_1+A_2)$. If we write~$A_1A_2^T=\exp(\theta\asymn)$ thanks to Rodrigues' formula~\eqref{eq:exponential_form} and we define~$A=A_2\exp(\frac12\theta\asymn)$, we get that~$A_1=A\exp(\frac12\theta\asymn)$ and so~~$\mathbb{M}=A\frac12(\exp(\frac12\theta\asymn)+\exp(-\frac12\theta\asymn))=A(\cos\frac{\theta}2\Id+(1-\cos\frac{\theta}2)\nvec\otimes\nvec)$, thanks to Rodrigues' formula~\eqref{eq:Rodrigues_formula} and to~\eqref{eq-asym2-proj}. Since the matrix~$S=\cos\frac{\theta}2\Id+(1-\cos\frac{\theta}2)\nvec\otimes\nvec$ is a positive-definite symmetric matrix as soon as~$\theta\in[0,\pi)$, we have that~$\det(\mathbb{M})> 0$. The polar decomposition of~$\mathbb{M}$ is then~$A$, which is the midpoint of the geodesic joining~$A_1$ to~$A_2$ (which corresponds to the curve~$t\in[0,\theta]\mapsto A_1\exp(t\asymn)$).

As soon as we average more than two matrices, there exist cases for which~$\det(\mathbb{M})<0$: for instance if we take 
\[A_1=\begin{pmatrix}1&0&0\\0&-1&0\\0&0&-1\end{pmatrix},A_2=\begin{pmatrix}-1&0&0\\0&1&0\\0&0&-1\end{pmatrix},A_3=\begin{pmatrix}-1&0&0\\0&-1&0\\0&0&1\end{pmatrix},\]
we have~$\mathbb{M}=\frac13(A_1+A_2+A_3)=-\frac{1}3\Id$.
	\end{enumerate}
	\end{remark}

\medskip	
	\textit{Noise term.} Agents make errors when trying to coordinate their body attitude with that of their neighbours. This is represented in the equation of~$A_k$ by a noise term:~$2\sqrt{D} dW^k_t$ where~$D>0$ and~$W^{k}_t = \lp W_t^{k,i,j}\rp_{i,j=1,2,3}$ are independent Gaussian distributions (Brownian motion). 

\medskip

From all these considerations, we obtain the IBM \begin{eqnarray} \label{eq:IBM}
	d\mathbf{X_k}(t) &=& v_0 A_k(t) \vezero dt,\\
	dA_k(t) &=& P_{T_{A_k}} \circ \left[ \nu(\PD(\mathbb{M}_k)\cdot A_k) \PD(\mathbb{M}_k) dt + 2\sqrt{D} dW^k_t \right], \label{eq:IBM2}
\end{eqnarray}
where the Stochastic Differential Equation is in  Stratonovich sense (see \cite{gardiner}). The projection~$P_{T_{A_k}}$ and the fact that we consider the SDE in Stratonovich sense ensures that the solution~$A_k(t)$ stays in~$SO(3)$. The normalization constant~$2\sqrt{D}$ ensures that the diffusion coefficient is exactly~$D$ : the law~$p$ of the underlying process given by~$dA_k=2\sqrt{D} P_{T_{A_k}} \circ dW^k_t$ satisfies~$\partial_tp=D\Delta_Ap$ where~$\Delta_A=\nabla_A\cdot\nabla_A$ is the Laplace-Beltrami operator on~$SO(3)$. Notice the factor~$2\sqrt{D}$ instead of the usual~$\sqrt{2D}$ which is encountered when considering diffusion process on manifolds isometrically embedded in the euclidean space~$\mathbb{R}^n$, because we are here considering~$SO(3)$ embedded in~$\mathcal{M}$ (isomorphic to~$\mathbb{R}^9$), but with the metric~\eqref{eq:dot_product_SO3}, which corresponds to the canonical metric of~$\mathbb{R}^9$ divided by a factor~$2$. We refer to the book~\cite{hsu2002stochastic} for more insight on such stochastic processes on manifolds.

\subsection{Mean-field limit}
\label{sec:mean_field_limit}
 \medskip

We assume that the kernel of influence~$K$ is Lipschitz, bounded, with the following properties:
\be \label{eq:properties_for_K}
K=K(|x|)\geq 0, \quad \int_{R^3}K(|x|)\,dx =1, \quad \int_{\R^3}|x|^2 K(|x|) \, dx<\infty.
\ee
In \cite{bolley2012mean} the mean-field limit is proven for the Vicsek model. Using the techniques there it is straightforward to see that for
$$\mathbb{M}(x,t):= \frac{1}{N}\sum_{i=1}^N K(\mathbf{X}_i- x) A_i$$
the law~$f^N=f^N(x,A,t)$ of the empirical measure associated to the Stratonovich Stochastic Differential Equation (SDE):
\begin{eqnarray} \label{eq:stratonovichSDE}
	d\mathbf{X_k}(t) &=&v_0 A_k(t) \vezero dt,\\
	dA_k(t) &=& P_{T_{A_k}} \circ \left[ \nu(\mathbb{M}(\mathbf{X_k},t)\cdot A_k) \mathbb{M}(\mathbf{X_k},t) dt + 2\sqrt{D} dW^k_t \right], \label{eq:stratonovichSDE_eqA}
\end{eqnarray}
converges weakly~$f^N \to f$ as~$N\to \infty$. The limit satisfies the kinetic equation:
\[\partial_t f + v_0 A \vezero \cdot \nabla_x f = D \Delta_A f - \nabla_A \cdot \lp F[f]f \rp ,\]
with 
\beqar
  F[f] &:=& \nu(\mathbb{M}_f\cdot A) P_{T_{A}} (\mathbb{M}_f),\\
  \mathbb{M}_f &=& \int_{\R^3 \times SO(3)}K(x-x')f(x',A',t) A' \, dA'dx'.
\eeqar

\bigskip

The equations we are dealing with \eqref{eq:IBM}-\eqref{eq:IBM2}, since we consider the Polar Decomposition of the averaged body attitude~$\mathbb{M}_k$, are slightly different from \eqref{eq:stratonovichSDE}-\eqref{eq:stratonovichSDE_eqA}, which would correspond to the modelling point of view of Eq. \eqref{eq:nopolar}. As a consequence, the corresponding coefficient of the SDE is not Lipschitz any more and the known results for existence of solutions and mean-field limit (see \cite[Theorem 1.4]{sznitman}) fail. More precisely, the problem arises when dealing with matrices with determinant zero; the orthogonal matrix of the Polar Decomposition is not uniquely defined for matrices with determinant zero and, otherwise,~$\PD(\mathbb{M}_k) = \mathbb{M}_k(\sqrt{\mathbb{M}_k^T \mathbb{M}_k})^{-1}$ (Lemma~\ref{lem:polar_decomposition}).

 A complete proof of the previous results in the case of Eq. \eqref{eq:IBM}-\eqref{eq:IBM2}  would involve proving that solutions to the equations stay away from the singular case~$\det(\mathbb{M}_k)=0$. This is an assumption that we make on the Individual Based Model (see the second point of Remark~\ref{rem:IBM}). This kind of analysis has been done for the Vicsek model  (explained in the introduction) in \cite{Figalli} where the authors prove global well-posedness for the kinetic equation in the spatially homogeneous case.
 
  In our case one expects the following to hold:
\begin{proposition}[Formal]
\label{prop:mean_field_limit}
When the number of agents in \eqref{eq:IBM}-\eqref{eq:IBM2}~$N\to \infty$,  its corresponding empirical distribution 
$$f^N(x,A,t) = \frac{1}{N}\sum_{k=1}^N \delta_{\lp X_k(t), A_k(t) \rp}$$
 converges weakly to~$f=f(x,A,t)$,~$(x,A,t)\in \mathbb{R}^3 \times SO(3)\times [0,\infty)$ satisfying
\begin{eqnarray} \label{eq:MFlimit_no_scaled}
	&&\partial_t f +  v_0 A \vezero\cdot \nabla_x f =  D\Delta_A f - \nabla_A \cdot \lp f F[f] \rp,\\
	&&F[f] := \nu P_{T_{A}} (\bar{\mathbb{M}}[f]), \nonumber\\
	&&\bar{\mathbb{M}}[f]= \PD(\mathbb{M}[f]), \quad \mathbb{M}[f](x, t) := \int_{\mathbb{R}^3 \times SO(3)} K\lp x-x' \rp f(x', A',t) A' dA' dx', \nonumber
\end{eqnarray}
 where~$\PD(\mathbb{M}[f])$ corresponds to the orthogonal matrix obtained on the Polar Decomposition of~$\mathbb{M}[f]$ (see Lemma~\ref{lem:polar_decomposition}); and~$\nu=\nu(\bar{\mathbb{M}}[f] \cdot A)$.
 
\end{proposition}

\section{Hydrodynamic limit}
\label{sec:macro}

The goal of this section will be to derive the macroscopic equations (Theorem~\ref{th:macro_limit}).
From now on, we consider the kinetic equation given in \eqref{eq:MFlimit_no_scaled}.

	\subsection{Scaling and expansion}
        \label{sec:scaling}
	
	We express the kinetic Eq. \eqref{eq:MFlimit_no_scaled} in dimensionless variables. Let~$\nu_0$ be the typical interaction frequency scale so that~$\nu(\bar{A} \cdot A) =\nu_0 \nu'(\bar{A} \cdot A)$ with~$\nu'(\bar{A} \cdot A) = \mathcal{O}(1)$. We introduce also the typical time and space scales~$t_0$,~$x_0$ such that~$t_0=\nu_0^{-1}$ and~$x_0=v_0t_0$; the associated variables will be~$t'=t/t_0$ and~$x'=x/x_0$. Consider the dimensionless diffusion coefficient~$d=D/\nu_0$ and the rescaled influence kernel~$K'(|x'|)=K(x_0|x'|)$. 
	Skipping the primes we get 
\begin{eqnarray*} 
	&&\partial_t f +  A \vezero\cdot \nabla_x f = d\Delta_A f - \nabla_A \cdot \lp f F[f] \rp, \\
	&&F[f] := \nu(\bar{\mathbb{M}}[f]\cdot A)\, P_{T_{A}} (\bar{\mathbb{M}}[f]),\\
	&&\bar{\mathbb{M}}[f]= \PD(\mathbb{M}[f]), \quad \mathbb{M}[f](x, t) := \int_{\mathbb{R}^3 \times SO(3)} K\lp x-x' \rp f(x', A',t) A' dA' dx'.
\end{eqnarray*}
 Here~$d$,~$\nu$ and~$K$ are assumed to be of order 1.
\begin{remark}
\label{rem:ratio_after_rescaling}
Notice in particular that before and after scaling the ratio
$$ \frac{\nu}{D}=\frac{\nu'}{d}$$
remains the same.
\end{remark}

	\medskip
	
Now, to carry out the macroscopic limit we rescale the space and time variables by setting~$\tilde t = \eps t$,~$\tilde x = \eps x$ to obtain (skipping the tildes):
\begin{align*}
	&\partial_t f^\eps +  A \vezero\cdot \nabla_x f^\eps = \frac{1}{\eps}\lp d\Delta_A f^\eps - \nabla_A \cdot \lp f^\eps F^\eps[f^\eps] \rp \rp,\\
	&F^\eps[f] := \nu (\bar{\mathbb{M}}^\eps[f]\cdot A)\,P_{T_{A}} (\bar{\mathbb{M}}^\eps[f^\eps]),\\
	&\bar{\mathbb{M}}^\eps[f]= \PD(\mathbb{M}^\eps[f]), \quad \mathbb{M}^\eps[f](x, t) := \int_{\mathbb{R}^3 \times SO(3)} K\lp \frac{x-x'}{\eps} \rp f(x', A',t) A' dA' dx'.
\end{align*}

\begin{lemma} \label{lem:expansion_operator}
Assuming that~$f$ is sufficiently smooth (with bounded derivatives), we have the expansion
$$\bar{\mathbb{M}}^\eps[f](x, t) = \Lambda[f](x,t) + \mathcal{O}(\eps^2),$$
 where
$$\Lambda[f](x,t) = \PD ( \lambda[f] ) \quad \mbox{and} \quad \lambda[f] =\int_{SO(3)} A' f(x,A',t) dA'.$$
\end{lemma} 
\begin{proof}
This is obtained by performing the change of variable~$x'=x+\eps\xi$ in the definition of~$\mathbb{M}^\eps[f]$ and using a Taylor expansion of~$f(x+\eps\xi,A',t)$ with respect to~$\eps$. We use that~$K$ is isotropic and with bounded second moment by assumption (see Eq.~\eqref{eq:properties_for_K}).
\end{proof} 
 
From the lemma, we rewrite
\begin{eqnarray} \label{eq:MFlimit}
	&&\partial_t f^\eps +  A \vezero\cdot \nabla_x f^\eps = \frac{1}{\eps} Q(f^\eps)+ \mathcal{O}(\eps), \\
	&&F_0[f] := \nu(\Lambda[f]\cdot A)\, P_{T_{A}} (\Lambda[f]), \nn\\
	&&\Lambda[f]= \PD(\lambda[f]), \quad \lambda[f](x, t) := \int_{SO(3)} f(x, A',t) A' dA', \nn \\
	&&Q(f) := d\Delta_A f - \nabla_A \cdot \lp f F_0[f] \rp.\nn
\end{eqnarray}
	
$\Lambda[f]$,~$Q(f)$ and~$F_0[f]$ are non-linear operators of~$f$, which only acts on the attitude variable~$A$.

\subsection{Preliminaries: differential calculus in~$SO(3)$}
\label{sec:differential_calculus}

In the sequel we will use the volume form, the gradient and divergence in~$SO(3)$ expressed in the Euler axis-angle coordinates~$(\theta,\nvec)$ (explained at the beginning of Section~\ref{sec:modelling}). In this section we give their explicit forms; the proofs are in appendix~\ref{ap:SO(3)}.

\begin{proposition}[The gradient in~$SO(3)$]
\label{prop:gradient_SO3}
Let~$f:SO(3) \to \R$ be a smooth scalar function. If~$\bar f (\theta, \nvec) = f(A(\theta,\nvec))$ is the expression of~$f$ in the Euler axis-angle coordinates by Rodrigues' formula~\eqref{eq:Rodrigues_formula}, we have then
\be \label{eq:gradient_SO3}
\nabla_A f = \partial_\theta \bar f\, A\asymn + \frac{1}{2\sin(\theta/2)} A\lp \cos(\theta/2)\left[\nabla_{\nvec}\bar f \right]_{\times} + \sin(\theta/2) \left[\nvec \times \nabla_{\nvec}\bar f \right]_\times \rp,
\ee
where~$A=A(\theta, \nvec)$ and~$\nabla_\nvec$ is the gradient on the sphere~$S^2$.
\end{proposition}

The volume form in~$SO(3)$ is left invariant (it is the Haar measure), due to the fact that the inner product in~$\mathcal{M}$ is also left invariant:~$A\cdot B=\frac12\tr(A^TB)=\Lambda A\cdot\Lambda B$ when~$\Lambda\in SO(3)$. We give its expression in the Euler axis-angle coordinates~$(\theta,\mathbf{n})$ :

\begin{lemma}[Decomposition of the volume form in~$SO(3)$]\label{lem-volume-form}
If~$\bar f (\theta, \nvec) = f(A(\theta,\nvec))$ is the expression of~$f$ in the Euler axis-angle coordinates by Rodrigues' formula~\eqref{eq:Rodrigues_formula}, we have 
$$\int_{SO(3)} f(A)\, dA = \int^\pi_0 W(\theta) \int_{S^2} \bar f(\theta, \nvec) \, d\nvec d\theta,~$$
where~$dn$ is the Lebesgue measure on the sphere~$S^2$, normalized to be a probability measure, and
\be \label{eq:volume_element}
W(\theta)= \frac{2}{\pi} \sin^2(\theta/2).
\ee
\end{lemma}

We have seen in Prop.~\ref{prop:gradient_SO3} that the gradient is decomposed in the basis~
$$\{A\asymn,A\left[\nabla_{\nvec}\bar f \right]_{\times}, A\left[\nvec \times \nabla_{\nvec}\bar f \right]_\times\},$$
 which are three orthogonal vectors of~$T_A$ (by Prop.~\ref{prop:tangentspaceSOn}).

More generally if~$B\in T_A$ for~$A=A(\theta,\nvec)\in SO(3)$, then~$B$ is of the form~$AH$ with~$H$ antisymmetric, so~$H=[\uu]_\times$ for some~$\uu\in\R^3$.  Decomposing~$\uu$ on~$\nvec$ and its orthogonal, we get that there exists~$\vv\perp \nvec$ and~$b\in\mathbb{R}$ such that~$B=bA\asymn+A \left[ \vv(\theta,\nvec)\right]_\times$. Expressing~$B$ in this form, we compute the divergence in~$SO(3)$.

\begin{proposition}[The divergence in~$SO(3)$]
\label{prop:divergence_SO3}
Consider~$B:SO_3\to T(SO(3))$ a smooth function (so that~$B(A)\in T_A$ for all~$A\in SO(3)$), and suppose that
$$B(A(\theta, \nvec)) = b(\theta,\nvec) A\asymn+ A \left[ \vv(\theta,\nvec)\right]_\times$$ for some smooth function~$b$ and smooth vector function~$\vv$ such that~$\vv(\theta,\nvec)\perp \nvec$.
Then
\beqar
\nabla_A \cdot  B &=& \frac{1}{\sin^2(\theta/2)} \partial_\theta \lp\sin^2(\theta/2) b (\theta, \nvec) \rp \\
&&\,+ \frac{1}{2\sin(\theta/2)} \nabla_\nvec \cdot \Big( \vv(\theta, \nvec) \cos(\theta/2) + \lp \vv(\theta, \nvec) \times \nvec \rp \sin(\theta/2) \Big).
\eeqar

\end{proposition}

\medskip

Now we can compute the Laplacian in~$SO(3)$:

\begin{corollary}\label{cor:laplacian_SO3}
The Laplacian in~$SO(3)$ can be expressed as
$$\Delta_A f = \frac{1}{\sin^2(\theta/2)} \partial_\theta\lp\sin^2(\theta/2) \partial_\theta \tilde f \rp + \frac{1}{4 \sin^2(\theta/2)} \Delta_\nvec \tilde f,$$
where~$\Delta_\nvec$ is the Laplacian on the sphere~$S^2$ and~$f(A)=f(A(\theta,\nvec))=\tilde{f}(\theta,\nvec)$.
\end{corollary}

\begin{proof}
 Let~$B(\theta,\nvec) :=\nabla_A f(A(\theta,\nvec))\in T_A$. Then, using the notations of Prop.~\ref{prop:divergence_SO3} and the result of Prop.~\ref{prop:gradient_SO3}, we have that
\beqar
b &=& \partial_\theta \tilde f,\\
\vv &=& \frac{1}{2 \sin(\theta/2)} \big( \cos(\theta/2) \nabla_\nvec \tilde f + \sin(\theta/2) (\nvec \times \nabla_\nvec \tilde f) \big),
\eeqar
from here we just need to apply Prop.~\ref{prop:divergence_SO3} knowing that~$(\nvec \times \nabla_\nvec \tilde f) \times \nvec = \nabla_\nvec \tilde f$ since~$\nabla_\nvec \tilde f$ is orthogonal to~$\nvec$. 
\end{proof}

\subsection{Equilibrium solutions and Fokker-Planck formulation}
\label{sec:FP}

We define a generalization of the von-Mises distributions on~$SO(3)$ by
\be	\label{eq:Von_Mises_equilibria}
	M_{\Lambda}(A)= \frac{1}{Z}\exp \lp \frac{\sigma(A \cdot \Lambda)}{d} \rp, \quad \int_{SO(3)}M_{\Lambda}(A)\,dA =1,\quad \Lambda\in SO(3), 
\ee
	where~$Z=Z(\nu,d)$ is a normalizing constant and~$\sigma=\sigma(\mu)$ is such that~$(d/d\mu)\sigma=\nu(\mu)$. Observe that~$Z<\infty$ is independent of~$\Lambda$ since the volume form on~$SO(3)$ is left-invariant. Therefore we have
\[Z=\int_{SO(3)}\exp(d^{-1}\sigma(A\cdot\Lambda))dA=\int_{SO(3)}\exp(d^{-1}\sigma(\Lambda^TA\cdot\Id))dA=\int_{SO(3)}\exp(d^{-1}\sigma(A\cdot\Id))dA,\]
and we also obtain that~$M_\Lambda(A)$ is actually~$M_\Id(\Lambda^TA)$.

	\medskip
We are now ready to describe the properties of~$Q$ in terms of these generalized von-Mises distributions.

	\begin{lemma}[Properties of~$Q$] \label{lem:propertiesQ}
The following holds:
	\begin{itemize}
		\item[i)] The operator~$Q$ can be written as
		$$Q(f) = d \nabla_A \cdot \left[ M_{\Lambda[f]} \nabla_A \lp \frac{f}{M_{\Lambda[f]}} \rp\right]$$
		and we have
\be		\label{eq:H}
		H(f) := \int_{SO(3)} Q(f) \frac{f}{M_{\Lambda[f]}} dA = -d \int_{SO(3)} M_{\Lambda[f]} \left|\nabla_A \lp \frac{f}{M_{\Lambda[f]}}\rp \right|^2 dA \leq 0.
		\ee
		\item[ii)] The equilibria, i.e., the functions~$f=f(x,A,t)$ such that~$Q(f)=0$ form a~$4$-dimensional manifold~$\mathcal{E}$ given by
		$$\mathcal{E}= \{ \rho M_{\Lambda}(A) \quad |\quad \rho >0, \quad \Lambda \in SO(3)\},$$
		where~$\rho$ is the total mass while~$\Lambda$ is mean body attitude of~$\rho M_{\Lambda}(A)$, i.e.,
		\begin{eqnarray*}
			\rho &=&\int_{SO(3)} \rho M_{\Lambda}(A) dA , \\
			\Lambda &=&\Lambda[\rho M_\Lambda].
		\end{eqnarray*}		 
	Furthermore,~$H(f)=0$ iff~$f=\rho M_\Lambda$ for arbitrary~$\rho \in \mathbb{R}_+$ and~$\Lambda \in SO(3)$.
	\end{itemize}
	\end{lemma}

	To prove  Lemma~\ref{lem:propertiesQ} we require the following one, which is of independent interest and for which we introduce the following notation: for any scalar function~$g:(0,\pi)\to \R$ and a given integrable scalar function~$h:(0,\pi)\to \R$ which remains positive (or negative) on~$(0,\pi)$, we define
	\be \label{eq:integral_notation}
	\la g(\theta)\ra_{h(\theta)}:= \int^\pi_0 g(\theta)\frac
	{h(\theta)}{\int^\pi_0 h(\theta')\, d\theta'}\, d\theta .
	\ee

\begin{lemma}[Consistency relation for the `flux'] \label{lem:consistency_relation}
	$$\lambda[M_{\Lambda_0}]=c_1\Lambda_0$$
	where~$c_1\in (0,1)$ is equal to
    \be \label{eq:c1}	
	c_1=\tfrac23\langle\tfrac12+\cos\theta\rangle_{m(\theta)\sin^2(\theta/2)} \, 
	\ee
	for 
\begin{equation}\label{eq-def-m}
m(\theta)= \exp( d^{-1}\sigma(\tfrac{1}{2}+\cos\theta)) .
\end{equation}
\end{lemma}

\begin{proof}
Using the fact that the measure on~$SO(3)$ is left invariant, we obtain
\begin{eqnarray*}
	\lambda[M_{\Lambda_0}]&=&\frac{1}{Z}\int_{SO(3)} A\exp(d^{-1}\sigma( (A\cdot \Lambda_0))) dA\\
	&=& \frac{\Lambda_0}{Z}\int_{SO(3)}  \Lambda_0^T A  \exp(d^{-1}\sigma(\tfrac12\tr(\Lambda_0^T A ))) dA\\
	&=& \frac{\Lambda_0}{Z}\int_{SO(3)}  B \exp(d^{-1}\sigma( \tfrac12\tr(B ))) dB.\\
\end{eqnarray*}

We now write~$B=\Id+\sin\theta\,[\nvec]_\times+(1-\cos\theta)[\nvec]_\times^2$ thanks to Rodrigues' formula~\eqref{eq:Rodrigues_formula}. Therefore, using Lemma~\ref{lem-volume-form}, we get
\begin{align*}
\lambda[M_{\Lambda_0}]&=\Lambda_0\frac{\int_{SO(3)}  B \exp(d^{-1}\sigma( \frac12\tr(B ))) dB}{\int_{SO(3)}  \exp(d^{-1}\sigma( \frac12\tr(B ))) dB}\\
&=\Lambda_0\frac{\int_0^\pi\sin^2(\theta/2) \exp(d^{-1}\sigma(\frac12+\cos\theta)) \big( \int_{S^2}(\Id+\sin\theta\,[\nvec]_\times+(1-\cos\theta)[\nvec]_\times^2) d\nvec \big) d\theta}{\int_0^\pi\sin^2(\theta/2) \exp(d^{-1}\sigma(\frac12+\cos\theta))d\theta}.
\end{align*}
Next, we see that since the function~$\nvec\mapsto[\nvec]_\times$ is odd, we have~$\int_{S^2}[\nvec]_\times d\nvec=0$. We also have (see~\eqref{eq-asym2-proj}) that~$[\nvec]_\times^2=\nvec\otimes\nvec-\Id$. Since we know that~$\int_{S^2}\nvec\otimes\nvec d\nvec=\frac13\Id$ (by invariance by rotation), it is easy to see that the integral in~$S^2$ has to be proportional to~$\Id$, the coefficient is given by computing the trace), we get that
\begin{align*}
\lambda[M_{\Lambda_0}]&=\Lambda_0\frac{\int_{SO(3)}  B \exp(d^{-1}\sigma( \frac12\tr(B ))) dB}{\int_{SO(3)} \exp(d^{-1}\sigma( \frac12\tr(B ))) dB}\\
&=\Lambda_0\frac{\int_0^\pi\sin^2(\theta/2) \exp(d^{-1}\sigma(\frac12+\cos\theta))(\Id+(1-\cos\theta)(\frac13-1)\Id) d\theta}{\int_0^\pi\sin^2(\theta/2) \exp(d^{-1}\sigma(\frac12+\cos\theta))d\theta}\\
&=\frac{\int_0^\pi\frac23(\frac12+\cos\theta)\sin^2(\theta/2) \exp(d^{-1}\sigma(\frac12+\cos\theta)) d\theta}{\int_0^\pi\sin^2(\theta/2) \exp(d^{-1}\sigma(\frac12+\cos\theta))d\theta}\Lambda_0=c_1 \Lambda_0,
\end{align*}
which gives the formula~\eqref{eq:c1} for~$c_1$. 

It remains to prove that~$c_1\in(0,1)$. We have that~$c_1$ is the average of~$\frac23(\frac12+\cos\theta)$ for the probability measure on~$(0,\pi)$ proportional to~$\sin^2(\theta/2) \exp(d^{-1}\sigma(\frac12+\cos\theta))$. Since we have~$\frac23(\frac12+\cos\theta)\leq 1$ with equality only for~$\theta=0$, we immediately get that~$c_1<1$. To prove the positivity, we remark that the function in the exponent~$\theta\mapsto d^{-1}\sigma(\frac12+\cos\theta)$ is strictly decreasing for~$\theta\in(0,\pi)$ (since~$\nu>0$ is the derivative of~$\sigma$), so we obtain that~$\sigma(\frac{1}{2}+\cos\theta)>\sigma(\frac{1}{2}+ \cos\frac{2\pi}{3})=\sigma(0)$ for~$\theta\in(0,\frac{2\pi}3)$. Therefore, for~$\theta\in(0,\frac{2\pi}3)$,
\[(\tfrac12+\cos\theta)\exp(d^{-1}\sigma(\tfrac12+\cos\theta))>(\tfrac12+\cos\theta)\exp(d^{-1}\sigma(0)),\]
since~$\frac12+\cos\theta>0$. When~$\theta\in(\frac{2\pi}3,\pi)$, we have exactly the same inequality above since we have~$\frac12+\cos\theta<0$.
Therefore we get
\[c_1>\frac{\int_0^\pi\frac23(\frac12+\cos\theta)\sin^2(\theta/2) \exp(d^{-1}\sigma(0)) d\theta}{\int_0^\pi\sin^2(\theta/2) \exp(d^{-1}\sigma(\frac12+\cos\theta))d\theta}=0,\]
since~$\int_0^\pi(\frac12+\cos\theta)\sin^2(\theta/2)d\theta=\int_0^\pi(\frac12+\cos\theta)(\frac12-\frac12\cos\theta)d\theta=\frac{\pi}4-\frac12\int_0^\pi\cos^2\theta d\theta=0$.
\end{proof}

\begin{proof}[Proof of Lemma~\ref{lem:propertiesQ}]
We follow the structure of the analogous proof in \cite{degond2008continuum}:
\begin{itemize}
	\item[i)] To prove the first identity we have that (see expression \eqref{eq:gradient_product})
	\beqar
	 \nabla_A \lp \ln M_{\Lambda[f]}\rp &=& d^{-1}\nabla_A \lp \sigma( A \cdot \Lambda[f])\rp\\
	 &=& d^{-1} \nu(A\cdot \Lambda[f])\,P_{T_A}\lp \Lambda[f]\rp\\
	 &=& d^{-1}F_0[f]
	\eeqar
	and so
	\beqar
	d \nabla_A \cdot \left[ M_{\Lambda[f]}\nabla_A \lp\frac{f}{M_{\Lambda[f]}} \rp\right] &=& d \nabla_A \cdot \left[\nabla_A f - f \nabla_A \lp\ln(M_{\Lambda[f]}) \rp \right]\\
	&=& d \nabla_A f -\nabla_A \cdot \lp fF_0[f] \rp.
	\eeqar
	Inequality \eqref{eq:H} follows from this last expression  and the Stokes theorem in~$SO(3)$.
	\item[ii)]  From the inequality \eqref{eq:H} we have that if~$Q(f)=0$, then~$\frac{f}{M_{\Lambda[f]}}$ is a constant that we denote by~$\rho$ (which is positive since~$f$ and~$M_{\Lambda[f]}$ are positive). Conversely, if~$f =\rho M_\Lambda$ then
	$$ \lambda[\rho M_\Lambda] = \int_{SO(3)} \rho M_\Lambda(A) A\, dA= \rho c_1 \Lambda$$
	by Lemma~\ref{lem:consistency_relation}. Now, by uniqueness of the Polar Decomposition and since~$\rho c_1 \Id$ is a symmetric positive-definite matrix, we have that~$\Lambda[\rho M_\Lambda]=\Lambda$.
\end{itemize}
\end{proof}	

Let us describe the behavior of these equilibrium distributions for small and large noise intensities. We have that for any function~$g$, the average~$\langle g(\tfrac12+\cos\theta)\rangle_{m(\theta)\sin^2(\theta/2)}$ is the average of~$g(A\cdot\Lambda)$ with respect to the probability measure~$M_\Lambda$ (by left-invariance, this is independent of~$\Lambda$).

One can actually check that the probability measure~$M_\Lambda$ on~$SO(3)$
converges in distribution to the uniform measure when~$d \rightarrow \infty$ (by Taylor expansion) and it converges to a Dirac delta at matrix~$\Lambda$ when~$d \rightarrow 0$ (this can be seen for~$M_\Id$ thanks to the decomposition of the volume form and the Laplace method, since the maximum of~$\sigma(\frac12+\cos\theta)$ is reached only at~$\theta=0$ which corresponds to the identity matrix, and we then get the result for any~$\Lambda$ since~$M_\Lambda(A)=M_\Id(\Lambda^TA)$). So for small diffusion, at equilibrium, agents tend to adopt the same body attitude close to~$\Lambda$.

With these asymptotic considerations, we have in particular the behaviour of~$c_1$ :
$$c_1\underset{d \rightarrow\infty}{\longrightarrow} 0$$
and
$$c_1 \underset{d\rightarrow 0}{\longrightarrow} 1.$$

	\subsection{Generalized Collision Invariants}
	\label{sec:GCI}

To obtain the macroscopic equation, we start by looking for the conserved quantities of the kinetic equation: we want to find the functions~$\psi=\psi(A)$ such that
$$\int_{SO(3)}Q(f) \psi \, dA=0 \qquad \mbox{for all } f.$$
By Lemma~\ref{lem:propertiesQ}, this can be rewritten as 
$$
0=-\int_{SO(3)}M_\Lambda[f] \nabla_A \lp\frac{f}{M_\Lambda[f]}	\rp \cdot \nabla_A \psi \, dA.
$$
 This happens if~$\nabla_A \psi \in T_{A}^\perp$ which holds true only if~$\nabla_A \psi =0$, implying that~$\psi$ is constant.

Consequently, our model has only one conserved quantity: the total mass. However the equilibria is~$4$-dimensional (by Lemma~\ref{lem:propertiesQ}). To obtain the macroscopic equations for~$\Lambda$, a priori we would need 3 more conserved quantities. This problem is sorted out by using Generalized Collision Invariants (GCI) a concept first introduced in \cite{degond2008continuum}.

\subsubsection{Definition and existence of GCI}

Define the operator
$$\mathcal{Q}(f, \Lambda_0) := \nabla_A \cdot \lp M_{\Lambda_0} \nabla_A \lp \frac{f}{M_{\Lambda_0}}\rp \rp,$$
notice in particular that
$$Q(f) = \mathcal{Q}(f, \Lambda[f]).$$
Using this operator we define:
\begin{definition}[Generalised Collision Invariant]\label{def-GCI}
For a given~$\Lambda_0 \in SO(3)$ we say that a real-valued function~$\psi: SO(3) \rightarrow \mathbb{R}$ is a Generalized Collision Invariant associated to~$\Lambda_0$, or for short~$\psi \in GCI(\Lambda_0)$, if
$$\int_{SO(3)} \mathcal{Q}(f, \Lambda_0) \psi\, dA =0 \quad \mbox{for all } f \mbox{ s.t } P_{T_{\Lambda_0}}(\lambda[f])=0.$$
\end{definition}
In particular, the result that we will use is :
\begin{equation}
\label{eq-property-GCI}
\psi\in GCI(\Lambda[f])\implies \int_{SO(3)} Q(f) \psi\, dA =0.
\end{equation}
Indeed, since~$\Lambda[f]$ is the polar decomposition of~$\lambda[f]$, we have~$\lambda[f]=\Lambda[f]S$, with~$S$ a symmetric matrix. Therefore (see Proposition~\ref{prop:tangentspaceSOn}), we get that~$\lambda[f]$ belongs to the orthogonal of~$T_{\Lambda[f]}$, so the definition~\ref{def-GCI} and the fact that~$Q(f) = \mathcal{Q}(f, \Lambda[f])$ gives us the property~\eqref{eq-property-GCI}.

The definition~\ref{def-GCI} is equivalent to the following:
\begin{proposition} \label{prop:equivalent_definition_GCI}
We have that~$\psi \in GCI(\Lambda_0)$ if and only if 
\begin{equation} \label{eq:for_psi}
\text{there exists }B\in T_{\Lambda_0}\text{ such that }\nabla_A \cdot (M_{\Lambda_0}\nabla_A \psi) = B \cdot A\, M_{\Lambda_0}.
\end{equation}

\end{proposition}

\begin{proof}[Proof of Prop.~\ref{prop:equivalent_definition_GCI}]
We denote~$\mathcal{L}$ the linear operator~$\mathcal{Q}(\cdot,\Lambda_0)$, and~$\mathcal{L}^*$ its adjoint. We have the following sequence of equivalences, starting from Def.~\ref{def-GCI}:
\begin{align*}
\psi\in GCI(\Lambda_0)&\Leftrightarrow \int_{SO(3)}\psi \mathcal{L}(f) \, dA =0, \quad \mbox{for all } f\mbox{ such that } P_{T_{\Lambda_0}}(\lambda[f])=0\\
&\Leftrightarrow  \int_{SO(3)}\mathcal{L}^*(\psi) f \, dA =0, \quad \mbox{for all } f\mbox{ such that } \int_{SO(3)} Af(A) dA \in \lp T_{\Lambda_0}\rp^\perp \\
&\Leftrightarrow  \int_{SO(3)}\mathcal{L}^*(\psi) f \, dA =0, \mbox{ for all } f\mbox{ s.t. } \forall B\in T_{\Lambda_0}, \int_{SO(3)} (B\cdot A)f(A) dA=0\\
&\Leftrightarrow  \int_{SO(3)}\mathcal{L}^*(\psi) f \, dA =0, \quad \mbox{for all } f\in F_{\Lambda_0}^\perp\\
&\Leftrightarrow  \mathcal{L}^*(\psi) \in \lp F^\perp_{\Lambda_0}\rp^\perp,
\end{align*}
where
$$F_{\Lambda_0}:= \left\{ g:SO(3)\to \R, \mbox{ with } g(A)=(B\cdot A), \, \mbox{ for some } B\in T_{\Lambda_{0}}\right\},$$
and~$F_{\Lambda_0}^\perp$ is the space orthogonal to~$F_{\Lambda_0}$ in~$L^2$.~$F_{\Lambda_0}$ is a vector space in~$L^2$ isomorphic to~$T_{\Lambda_0}$ and~$(F_{\Lambda_0}^\perp)^\perp=F_{\Lambda_0}$ since~$F_{\Lambda_0}$ is closed (finite dimensional). Therefore we get
\[\psi\in GCI(\Lambda_0)\Leftrightarrow  \mathcal{L}^*(\psi) \in  F_{\Lambda_0}\Leftrightarrow  \mbox{ there exists } B \in T_{\Lambda_0}\mbox{ such that } \mathcal{L}^*(\psi)(A)= B\cdot A ,
\]
which ends the proof since the expression of the adjoint is~$\mathcal{L}^*(\psi)=\frac1{M_{\Lambda_0}}\nabla_A \cdot (M_{\Lambda_0}\nabla_A \psi)$.
\end{proof}

We prove the existence and uniqueness of the solution~$\psi$ satisfying Eq. \eqref{eq:for_psi} in the following:
\begin{proposition}[Existence of the GCI]
\label{prop:existence_GCI}
For a given~$B\in T_\Lambda$ fixed, there exists a unique (up to a constant)~$\psi_B\in H^1(SO(3))$, satisfying the relation \eqref{eq:for_psi}.
\end{proposition}
\begin{proof}[Proof of Prop.~\ref{prop:existence_GCI}]
We would like to apply the Lax-Milgram theorem to prove the existence of~$\psi$ in an appropriate functional space. For this, we rewrite the relation \eqref{eq:for_psi} weakly
\be \label{eq:definition_a_b}
a(\psi,\varphi) := \int_{SO(3)}M_{\Lambda_0}\nabla_A \psi \cdot \nabla_A \varphi\, dA = \int_{SO(3)} B \cdot P_{T_{\Lambda_0}}(A) M_{\Lambda_0} \varphi\, dA=:b(\varphi).
\ee
Our goal is to prove that there exists a unique~$\psi\in H^1(SO(3))$ such that
$a(\psi,\varphi)=b(\varphi)$ for all~$\varphi\in H^1(SO(3))$.

To begin with we apply the Lax-Milgram theorem on the space
$$H^1_0(SO(3)):=\left\{\varphi \in H^1 \, |\, \int_{SO(3)}\varphi \, dA =0 \right\}.$$
In this space the~$H^1$-norm and the~$H^1$ semi-norm are equivalent thanks to the Poincar\'e inequality, i.e., there exists~$C>0$ such that
$$ \int_{SO(3)}|\nabla_A\varphi|^2 dA \geq C \int_{SO(3)} |\varphi|^2 dA   \quad \mbox{for some } C>0, \quad \mbox{for all } \varphi \in H^1_0(SO(3)).$$
Notice that the Poincar\'e inequality holds in~$SO(3)$ because it is compact Riemannian manifold \cite{colbois}. 
This gives us the coercivity estimate to apply the Lax-Milgram theorem. Hence, there exists a unique~$\psi\in H^1_0(SO(3))$ s.t~$a(\psi, \varphi)= b(\varphi)$ for all~$\varphi \in H^1_0(SO(3))$.

\medskip
Now, define  for a given~$\varphi \in H^1(SO(3))$,~$\varphi_0:=\varphi - \int_{SO(3)}\varphi \,dA \in H^1_0(SO(3))$. It holds that
$$a(\psi, \varphi)= a(\psi, \varphi_0) \quad \mbox{ and }\quad b(\varphi)=b(\varphi_0)$$
since~$b(1)=0$ given that it has antisymmetric integrand.
Hence, we obtain that there exists a unique~$\psi \in H^1_0(SO(3))$ such that
$$a(\psi, \varphi)=b(\varphi) \quad \mbox{for all } \varphi \in H^1(SO(3)).$$

Suppose next, that there exists another solution~$\bar \psi \in H^1(SO(3))$ to this problem, then the difference~$\Psi=\psi-\bar \psi$ satisfies:
$$0=a(\Psi, \varphi)= \int_{SO(3)} M_{\Lambda_0} \nabla_A \Psi \cdot \nabla_A \varphi \, dA \qquad \mbox{ for all } \varphi \in H^1(SO(3)).$$
Take in particular~$\varphi = \Psi$, then
$$\int_{SO(3)} M_{\Lambda_0} |\nabla_A \Psi|^2 \, dA =0.$$
Hence,~$\Psi = c$ for some constant~$c$, so all solutions are of the form~$\psi +c$ where~$\psi$ is the unique solution satisfying~$\int_{SO(3)}\psi \, dA =0$. 
\end{proof}

By writing that
\be \label{eq:aux_antisymmetric}
B\in T_{\Lambda_0} \mbox{ if and only if there exists } P\in \mathcal{A}, \, B=\Lambda_0 P,
\ee
with~$\mathcal{A}$ the set of antisymmetric matrices, we deduce the:
\begin{corollary}
For a given~$\Lambda_0 \in SO(3)$, the set of Generalized Collision Invariants associated to~$\Lambda_0$ are
$$GCI(\Lambda_0) =span\{1, \cup_{P\in \mathcal{A}} \psi^{\Lambda_0}_P \}$$
(where~$\mathcal{A}$ is the set of antisymmetric matrices) 
with~$\psi^{\Lambda_0}_P$ the unique solution in~$H^1_0(SO(3))$ of
$$a(\psi^{\Lambda_0}_P, \varphi) = b_P(\varphi) \qquad\mbox{for all } \varphi \in H^1(SO(3)),$$
where~$a$ and~$b_P$ are defined by \eqref{eq:definition_a_b} with~$B$ substituted by~$\Lambda_0 P$.
\end{corollary}

Notice that since the mapping~$P\mapsto\psi^{\Lambda_0}_P$ is linear and injective from~$\mathcal{A}$ (of dimension~$3$) to~$H^1_0(SO(3))$, the vector space~$GCI(\Lambda_0)$ is of dimension~$4$.

\subsubsection{The non-constant GCIs}

From now on, we omit the subscript on~$\Lambda_0$, and we are interested in a simpler expression for~$\psi^{\Lambda}_P$.
Rewriting expression \eqref{eq:for_psi} using \eqref{eq:aux_antisymmetric}, for any given~$P\in \mathcal{A}$ we want to find~$\psi$ such that 
\begin{equation} \label{eq:strategy_GCI_condition}
\nabla_A \cdot (M_\Lambda \nabla_A \psi) = (\Lambda P) \cdot A\, M_\Lambda = P \cdot (\Lambda^T A) M_\Lambda, \quad P\in\mathcal{A}. 
\end{equation}

\begin{proposition}
Let~$P\in \mathcal{A}$ and~$\psi$ be the solution of \eqref{eq:strategy_GCI_condition} belonging to~$H^1_0(SO(3))$. If we denote~$\bar \psi (B) :=\psi(\Lambda B)$, then~$\bar \psi$ is the unique solution in~$H^1_0(SO(3))$ of:
\be \label{eq:rescaled_GCI}
\nabla_B \cdot \lp M_{\Id}(B)\nabla_B \bar \psi \rp = P \cdot B M_{\Id}(B).
\ee
\end{proposition}
\begin{proof}
Let~$\psi(A)= \bar \psi(\Lambda^T A)$. Consider~$A(\eps)$ a differentiable curve in~$SO(3)$ with 
$$A(0)=A, \qquad \left.\frac{dA(\eps)}{d\eps}\right|_{\eps=0}= \delta_A \in T_A.$$
Then, by definition
$$\lim_{\eps \to 0}\frac{\psi(A(\eps))-\psi(A)}{\eps}= \nabla_A\psi(A) \cdot \delta_A,$$
and therefore we have that
$$\lim_{\eps \to 0}\frac{\bar\psi(\Lambda^T A(\eps)) -\bar \psi(\Lambda^T A)}{\eps}= \nabla_B \bar \psi(\Lambda^T A) \cdot \Lambda^T\delta_A$$
since
$$\Lambda^T A(0)=\Lambda^TA, \qquad \left.\frac{d}{d\eps}\Lambda^TA(\eps)\right|_{\eps=0}=\Lambda^T \delta_A.$$
We conclude that 
$$\nabla_A\psi(A) \cdot \delta_A= \nabla_B \bar \psi(\Lambda^T A) \cdot \Lambda^T\delta_A.$$
Now we check that
\begin{eqnarray*}
\frac12\tr\lp \lp\nabla_A\psi(A)\rp^T \delta_A\rp &=&\frac12\tr\lp \lp \nabla_B \bar \psi(\Lambda^T A)\rp^T \Lambda^T \delta_A\rp\\
&=&\frac12\tr\lp \lp \Lambda\nabla_B \bar \psi(\Lambda^T A)\rp^T  \delta_A\rp,
\end{eqnarray*}
implying (since this is true for any~$\delta_A \in T_A$) that
$$\nabla_A\psi(A) = \Lambda \nabla_B \bar \psi(\Lambda^T A) .$$

\medskip
Now to deal with the divergence term, we consider the variational formulation. Consider~$\varphi\in H^1(SO(3))$, then our equation is equivalent to 
$$-\int_{SO(3)}M_\Lambda(A) \nabla_A \psi(A) \cdot \nabla_A \varphi(A) \, dA= \int_{SO(3)}P\cdot (\Lambda^T A) M_\Lambda(A) \varphi(A) \, dA$$
for all~$\varphi\in H^1(SO(3)).$
The left hand side can be written as:
\[\begin{split}
&-\int_{SO(3)} M_{\Id}(B)(\Lambda^TA)\lp \Lambda \nabla_B\bar\psi(\Lambda^TA) \rp \cdot \lp \Lambda \nabla_B\bar\varphi(\Lambda^TA) \rp dA \\
&=-\int_{SO(3)} M_{\Id}(B) \nabla_B \bar \psi(B) \cdot \nabla_B \bar\varphi(B)\, dB;
\end{split}\]
and the right hand side is equal to 
$$\int_{SO(3)}P\cdot B\, M_{\Id}(B) \bar \varphi(B) \, dB,$$
where we define analogously~$\bar \varphi(B)=\varphi(\Lambda B)$. This concludes the proof.
\end{proof}

Therefore it is enough to find the solution to \eqref{eq:rescaled_GCI}. Inspired by  \cite{degond2008continuum} we make the ansatz:
$$\bar \psi(B) = P\cdot B\, \bar \psi_0(\tfrac12\tr(B))$$
for some scalar function~$\bar \psi_0$.

\begin{proposition}[Non-constant GCI]
\label{prop:non_constant_GCI}
Let~$P\in \mathcal{A}$, then the unique solution~$\bar\psi \in H^1_0(SO(3))$ of \eqref{eq:rescaled_GCI} is given by
\be \label{eq:GCI_explicit}
\bar \psi (B) = P \cdot B \, \bar\psi_0 (\tfrac12\tr(B)),
\ee
where~$\bar \psi_0$ is constructed as follows:
Let~$\widetilde\psi_0: \R\to \R$ be the unique solution to
\be \label{eq:psi_0}
\frac{1}{\sin^2(\theta/2)} \partial_\theta\lp \sin^2(\theta/2) m(\theta)\partial_\theta \lp \sin\theta\widetilde\psi_0 \rp \rp - \frac{m(\theta)\sin\theta}{2\sin^2(\theta/2)}\widetilde\psi_0 = \sin\theta \,m(\theta),
\ee
where~$m(\theta)=M_{\Id}(B)=\exp(d^{-1}\sigma(\frac12+\cos\theta))/Z$. Then~
\be \label{eq:change_variable_psi0}
\widetilde \psi_0(\theta)= \bar\psi_0 \lp \frac12\tr(B) \rp
\ee
 by the relation~$\frac12\tr(B)= \frac12+\cos\theta$.~$\widetilde\psi_0$ is~$2\pi$-periodic, even and negative (by the maximum principle).

Going back to the GCI~$\psi(A)$, we can write it as 
\be  \label{eq:GCI_explicit_not_rescaled}
\psi(A) = P\cdot (\Lambda^T A)\, \bar{\psi}_0(\Lambda \cdot A).
\ee
\end{proposition}

\begin{proof}[Proof of Prop.~\ref{prop:non_constant_GCI}]
Suppose that the solution is given by expression \eqref{eq:GCI_explicit}. We check that~$\widetilde \psi_0$ given by Eq. \eqref{eq:change_variable_psi0} satisfies Eq. \eqref{eq:psi_0} using the gradient and divergence in~$SO(3)$ computed in Prop.~\ref{prop:gradient_SO3} and~\ref{prop:divergence_SO3}. First notice that~$P$ is antisymmetric, thus if we write  Rodrigues' formula~\eqref{eq:Rodrigues_formula} for~$B(\theta,\nvec)$, the symmetric part of~$B(\theta,\nvec)$ gives no contribution when computing~$P\cdot B$ and we get
$$\bar{\psi}(B) = P\cdot B\, \bar{\psi}_0(\tfrac12\tr(B)) = \sin\theta\,\widetilde{\psi}_0(\theta)P\cdot \asymn =\sin\theta\,\widetilde{\psi}_0(\theta) (\mathbf{p} \cdot \nvec),$$
 where the vector~$\mathbf{p}$ is such that~$P=[\mathbf{p}]_\times$ and this leads to
\beqar
\nabla_B\cdot \lp M_{\Id}(B)\nabla_{B}\bar{\psi} \rp & =& \frac{1}{\sin^2(\theta/2)}\partial_\theta\Big(\sin^2(\theta/2)m(\theta)\partial_\theta\lp \sin\theta\,\widetilde \psi_0(\theta) \rp \Big)(\mathbf{p}\cdot \nvec) \\
&& \quad +\,  \frac{m(\theta)\sin\theta}{4\sin^2(\theta/2)}\widetilde\psi_0(\theta)\Delta_n (\mathbf{p}\cdot \nvec).
\eeqar
Using that the Laplacian in the sphere has the property~$\Delta_\nvec (\mathbf{p} \cdot \nvec) = -2(\mathbf{p} \cdot \nvec)$ ($\mathbf{p}\cdot \nvec$ corresponds to the first spherical harmonic), we conclude that  expression \eqref{eq:psi_0} is satisfied. In the computation we used the same procedure as for the proof of the expression of the Laplacian in~$SO(3)$ (Corollary~\ref{cor:laplacian_SO3}), but (using the same notations) we have taken~$b(\theta,\nvec)= m(\theta) \partial_\theta(\sin\theta\, \widetilde{\psi}_0(\theta))(\mathbf{p}\cdot\nvec)$.

\medskip
To conclude the proof we just need to check that~$\widetilde \psi_0$ exists and corresponds to a function~$\bar\psi$ in~$H^1_0(SO(3))$. Using the expression of the volume form, since~$\int_{S^2}\mathbf{p}\cdot\nvec \,d \nvec=0$, we get that if~$\psi_0$ is smooth, we have~$\int_{SO(3)}\bar\psi(A)dA=0$, and using the expression of the gradient, we get that
\[\begin{split}
\int_{SO(3)}|\nabla\bar\psi(A)|^2dA=&\frac2{\pi}\int_0^\pi\sin^2(\theta/2)|\partial_\theta(\sin\theta\widetilde\psi_0(\theta))|^2d\theta\int_{S^2}|\mathbf{p}\cdot\nvec|^2d\nvec\\
&+\frac2{\pi}\int_0^\pi\frac14|\sin\theta\widetilde\psi_0(\theta)|^2d\theta\int_{S^2}|\nabla_\nvec(\mathbf{p}\cdot\nvec)|^2d\nvec\, .
\end{split}\]

Therefore by density of smooth functions in~$H^1_0(SO(3))$, we get that~$\bar\psi\in H^1_0(SO(3))$ if and only if~$\widetilde \psi_0\in H$, where 
$$H:= \left\{ \psi \, | \int_{(0,\pi)}\psi^2\sin^2\theta\, d\theta<\infty, \quad \int_{(0,\pi)}|\partial_\theta( \sin\theta\psi(\theta))|^2 \sin^2(\theta/2) \, d\theta <\infty \right\}.$$
This Hilbert space is equipped with the corresponding norm:
$$\|\psi\|^2_H = \int_{(0,\pi)}\psi^2\sin^2\theta\, d\theta +  \int_{(0,\pi)}|\partial_\theta (\sin\theta\psi(\theta))|^2 \sin^2(\theta/2) \, d\theta. ~$$

Now, Eq. \eqref{eq:psi_0} written in weak form in~$H$ and tested against any~$\phi \in H$ reads
\beqar
a(\widetilde\psi_0,\phi)&:=&-\int_{(0,\pi)}m(\theta)\left[\sin^2(\theta/2)\partial_\theta(\sin\theta\widetilde\psi_0(\theta))\partial_\theta(\sin\theta\phi(\theta))\, d\theta +\frac12\sin^2\theta\widetilde\psi_0(\theta)\phi(\theta)\right]\, d\theta \\
&=& \int_{(0,\pi)}\sin^2\theta\sin^2(\theta/2)m(\theta) \phi\, d\theta =:b(\phi).
\eeqar
It holds for some~$c,c',c''>0$ that:~$|a(\psi,\phi)| \leq  c\|\psi\|_H \|\phi\|_H$ since~$m=m(\theta)$ is bounded; and also~$|a(\psi,\psi)| \geq  c' \|\psi\|^2_{H}$ since there exists~$m_0>0$ such that~$m(\theta)>m_0$ for all~$\theta\in[0,\pi]$; finally, we also have that~$|b(\phi)| \leq c''\|\phi\|_H^2$. Therefore, by the Lax-Milgram theorem, there exists a (unique) solution~$\widetilde\psi_0 \in H$ to \eqref{eq:psi_0}, which corresponds to a (unique)~$\bar\psi$ in~$H^1_0(SO(3))$.
\end{proof}

 	\subsection{The macroscopic limit}
\label{sec:macro-limit}
In this section we investigate the hydrodynamic limit. To state the theorem we first give the definitions of the first order operators~$\delta_x$ and~$\rvec_x$. For a smooth function~$\Lambda$ from~$\R^3$ to~$SO(3)$, and for~$x\in\R^3$, we define the following matrix~$\mathcal{D}_x(\Lambda)$ such that for any~$\ww\in\R^3$, we have
\begin{equation}\label{eq-def-D}
(\ww\cdot\nabla_x)\Lambda=[\mathcal{D}_x(\Lambda)\ww]_\times\Lambda.
\end{equation}
Notice that this first-order differential equation~$\mathcal{D}_x$ is well-defined as a matrix; for a given vector~$\ww$, the matrix~$(\ww\cdot \nabla_x) \Lambda$ is in~$T_\Lambda$ and thanks to Prop.~\ref{prop:projection_tangentspace}, it is of the form~$P\Lambda$, with $P$ an antisymmetric matrix.   
Therefore there exists a vector~$\mathcal{D}_x(\Lambda)(\ww)\in \R^3$ depending on~$\ww$ such that~$P=[\mathcal{D}_x(\Lambda)(\ww)]_\times.$
The function~$\ww \mapsto \mathcal{D}_x(\Lambda)(\ww)$ is linear from~$\R^3$ to~$\R^3$, so~$\mathcal{D}_x(\Lambda)$ can be identified as a matrix. 

We now define the first order operators~$\delta_x$ (scalar) and~$\rvec_x$ (vector), by
\begin{equation}\label{def-delta-r}
\delta_x(\Lambda)=\tr\big(\mathcal{D}_x(\Lambda)\big) \quad\text{ and }\quad [\rvec_x(\Lambda)]_\times=\mathcal{D}_x(\Lambda)-\mathcal{D}_x(\Lambda)^T.
\end{equation}

We first give an invariance property which allows for a simple expression for these operators.
\begin{proposition}\label{prop:interpretation_delta_r} 
The operators~$\mathcal{D}_x$,~$\delta_x$ and~$\rvec_x$ are right invariant in the following sense: if~$A$ is a fixed matrix in~$SO(3)$ and $\Lambda:\R^3\to SO(3)$ a smooth function, we have

\begin{equation*}
\mathcal{D}_x(\Lambda A)=\mathcal{D}_x(\Lambda),\quad\delta_x(\Lambda A)= \delta_x(\Lambda) \quad \text{and} \quad \rvec_x(\Lambda A)= \rvec_x(\Lambda).
\end{equation*}

Consequently, in the neighborhood of~$x_0\in\R^3$, we can write~$\Lambda(x) = \exp\lp [\mathbf{b}(x)]_\times\rp\,\Lambda(x_0)$ where~$\mathbf{b}$ is a smooth function from a neighborhood of~$x_0$ into~$\R^3$ such that~$\mathbf{b}(x_0)=0$, and we have
\[\big(\mathcal{D}_x(\Lambda)\big)_{ij}(x_0)=\partial_j\mathbf{b}_i(x_0),\]
and therefore
\[ \delta_x(\Lambda)(x_0) =(\nabla_x \cdot\mathbf{b})\, (x_0),\quad \text{ and }\quad\rvec_\Lambda(x_0) = (\nabla_x \times \mathbf{b})\, (x_0).\]
\end{proposition}

\begin{proof}
For any~$\ww\in\R^3$, we have, since~$A$ is constant:
\[[\mathcal{D}_x(\Lambda A)\ww]_\times\Lambda A=\ww\cdot\nabla_x(\Lambda A)=(\ww\cdot\nabla_x\Lambda) A=[\mathcal{D}_x(\Lambda)\ww]_\times\Lambda A.\]
This proves that~$\mathcal{D}_x(\Lambda A)=\mathcal{D}_x(\Lambda)$, and by~\eqref{def-delta-r}, the same is obviously true for~$\delta_x$ and~$\rvec_x$.  

We now write, in the neighborhood of~$x_0$, that~$\Lambda(x)=\exp([\mathbf{b(x)}]_\times)\Lambda(x_0)$, with~$\mathbf{b}$ smooth in the neighborhood of~$x_0$ and~$\mathbf{b}(x_0)=0$. Then we have~$\mathcal{D}_x(\Lambda)=\mathcal{D}_x\big(\exp([\mathbf{b}]_\times)\big)$. We perform a Taylor expansion around~$x_0$ of~$\exp([\mathbf{b}]_\times)$:
\[\exp([\mathbf{b}(x)]_\times) = \Id + [\mathbf{b}(x)]_\times + M(x),\]
where~$M(x)$ is of order 2 in the coordinates~$\mathbf{b}_1$,~$\mathbf{b}_2$,~$\mathbf{b}_3$, (since~$\mathbf{b}$ is smooth in the neighborhood of~$x_0$ and~$\mathbf{b}(x_0)=0$), therefore
$$\partial_1 M(x_0) = \partial_2 M(x_0) = \partial_3 M(x_0)=0.$$
We then get, since~$\exp([\mathbf{b}(x_0)])=\Id$, that
\[[\mathcal{D}_x\big(\exp([\mathbf{b}]_\times)\big)(x_0)\ww]_\times=\ww\cdot\nabla_x\big(\exp([\mathbf{b}]_\times)\big)(x_0)=\big[(\ww\cdot\nabla_x\mathbf{b})(x_0)\big]_\times, \]
and therefore $\mathcal{D}_x(\Lambda)(x_0)\ww=\mathcal{D}_x(\exp([\mathbf{b}]_\times))(x_0)\ww=(\ww\cdot\nabla_x\mathbf{b})(x_0)$. Taking~$\ww=\mathbf{e}_j$, we get~$\mathcal{D}_x(\Lambda)(x_0)\mathbf{e}_j=\partial_j\mathbf{b}(x_0)$, and thus~$\big(\mathcal{D}_x(\Lambda)(x_0)\big)_{ij}=\mathbf{e}_i\cdot\mathcal{D}_x(\Lambda)(x_0)\mathbf{e}_j=\partial_j\mathbf{b}_i$. The formula for~$\delta_x(\Lambda)$ follows from~\eqref{def-delta-r}, since~$\nabla_x\cdot\mathbf{b}=\sum_i\partial_i\mathbf{b}_i$. Finally by the definition of $[\cdot]_\times$ (see~\eqref{eq:def_operator_asym}), we get
\[[\nabla_x\times\mathbf{b}]_\times=\begin{pmatrix}
0&\partial_2\mathbf{b}_1-\partial_1\mathbf{b}_2 & \partial_3\mathbf{b}_1-\partial_1\mathbf{b}_3\\
\partial_1\mathbf{b}_2-\partial_2\mathbf{b}_1 &0&\partial_3\mathbf{b}_2-\partial_2\mathbf{b}_3 \\ 
\partial_1\mathbf{b}_3-\partial_3\mathbf{b}_1& \partial_2\mathbf{b}_3-\partial_3\mathbf{b}_2&0
\end{pmatrix},\] 
so from~\eqref{def-delta-r} we obtain~$(\nabla_x\times\mathbf{b})(x_0)=\rvec_x(\Lambda)(x_0)$.
\end{proof}

We are now ready to state the main theorem of our paper (see Section~\ref{sec:discussion_macro} for a discussion on this result).

\begin{theorem}[(Formal) macroscopic limit]
\label{th:macro_limit}
When~$\eps \to 0$ in the kinetic Eq. \eqref{eq:MFlimit} it holds (formally) that
$$f_\eps \to f=f(x,A,t)=\rho M_\Lambda(A),\qquad \Lambda=\Lambda(t,x)\in SO(3),\, \rho=\rho(t,x)\geq 0.$$
Moreover, if this convergence is strong enough and the functions~$\Lambda$ and~$\rho$ are smooth enough, they satisfy the following first-order system of partial differential equations:
\begin{align}
&\partial_t \rho + \nabla_x \cdot \big( c_1 \rho \Lambda \mathbf{e_1}\big)=0, \label{eq:macro_rho}\\
&\rho \Big( \partial_t\Lambda+ c_2 \big((\Lambda \vezero) \cdot \nabla_x\big)\Lambda \Big) +\left[ (\Lambda \vezero) \times \big(c_3\nabla_x \rho+c_4\rho\,\rvec_x(\Lambda)\big) + c_4\rho\,\delta_x(\Lambda)\Lambda \vezero\right]_\times \Lambda =0,\label{eq:macro_lambda}
\end{align}
where~$c_1=c_1(\nu, d)=\tfrac23\langle\tfrac12+\cos\theta\rangle_{m(\theta)\sin^2(\theta/2)}$ is the constant given in \eqref{eq:c1} and
\beqar
c_2 &=&  \tfrac{1}{5}\langle 2+3\cos\theta\rangle_{\widetilde{m}(\theta)\sin^2(\theta/2)} ,\\
c_3 &=& d \langle \nu(\tfrac12+\cos\theta)^{-1} \rangle_{\widetilde{m}(\theta)\sin^2(\theta/2)},\\
c_4&=& \tfrac{1}{5}\langle 1-\cos\theta\rangle_{\widetilde{m}(\theta)\sin^2(\theta/2)},
\eeqar
where the notation
$\langle \cdot \rangle_{\widetilde{m}(\theta)\sin^2(\theta/2)}$ is defined in \eqref{eq:integral_notation}.
The function~$\widetilde{m}:(0,\pi)\to(0,+\infty)$ is given by
\be \label{eq:weight_mtilde}
\widetilde{m}(\theta):= \nu(\tfrac12+\cos\theta)\,\sin^2\theta\, m(\theta)\, \widetilde\psi_0(\theta),
\ee
where~$m(\theta)=\exp(d^{-1}\sigma(\tfrac{1}{2}+\cos\theta))$ is the same as in~\eqref{eq-def-m} and~$\widetilde \psi_0$ is the solution of Eq. \eqref{eq:psi_0}. 
\end{theorem} 
\begin{proof}
Suppose that~$f_\eps \rightarrow f$ as~$\eps \rightarrow 0$, then using~\eqref{eq:MFlimit} we get~$Q(f_\eps) = \mathcal{O}(\eps)$, which formally yields~$Q(f)=0$ and by Lemma~\ref{lem:propertiesQ} we have that
$$f=f(x,A,t)=\rho M_\Lambda(A),\quad \text{with}\quad \Lambda=\Lambda(t,x)\in SO(3),\, \rho=\rho(t,x)\geq 0.$$

\medskip
Using the conservation of mass (integrating~\eqref{eq:MFlimit} on~$SO(3)$), we have that
$$\partial_t \rho_\eps + \nabla_x \cdot j[f_{\eps}] = \mathcal{O}(\eps), $$
where
$$\rho_\eps(t,x):= \int_{SO(3)}f_\eps(x,A,t)\, dA, \quad j[f_\eps]:=\int_{SO(3)} A\vezero f_\eps \, dA,$$
and in the limit (formally)
$$\rho_\eps \rightarrow \rho,$$
$$j[\fe] \rightarrow \rho \int_{SO(3)}A\vezero M_\Lambda(A)\, dA=\rho\lambda[M_{\Lambda}]\vezero =\rho c_1 \Lambda \vezero,$$
thanks to Lemma~\ref{lem:consistency_relation}. 
This gives us the continuity equation~\eqref{eq:macro_rho} for~$\rho$.

\medskip
Now, we want to obtain the equation for~$\Lambda$. We write~$\Lambda^\eps=\Lambda[f^\eps]$, and we take~$P\in\mathcal{A}$ a given antisymmetric matrix. We consider the non-constant GCI associated to~$\Lambda^\eps$ and corresponding to~$P$ in~\eqref{eq:GCI_explicit_not_rescaled}:~$\psi^\eps(A)=P\cdot((\Lambda^\eps)^TA)\bar\psi_0(\Lambda^\eps\cdot A)$. Since we have~$\psi^\eps\in GCI(\Lambda[f^\eps])$, we obtain, thanks to the main property~\eqref{eq-property-GCI} of the GCI, that
\[\int_{SO(3)}Q(f^\eps)\psi^\eps dA=0.\]
Multiplying~\eqref{eq:MFlimit} by~$\psi^\eps$, integrating w.r.t.~$A$ on~$SO(3)$ and using the expression of~$\psi^\eps$ as stated above, we obtain
\[\int_{SO(3)} \big(\partial_t f^\eps+A \vezero \cdot \nabla_xf^\eps + \mathcal{O}(\eps) \big) P \cdot \big((\Lambda^\eps)^TA\big)\, \bar \psi_0(\Lambda^\eps \cdot A)\, dA = 0.
\]
Assuming the convergence~$f^\eps\to f$ is sufficiently strong, we get in the limit
\be \label{eq:applied_GCI}
\int_{SO(3)} \big(\partial_t (\rho M_\Lambda)+A \vezero \cdot \nabla_x(\rho M_\Lambda) \big) \big( P \cdot \Lambda^TA\big)\, \bar \psi_0(\Lambda \cdot A)\, dA =0.
\ee
Since \eqref{eq:applied_GCI} is true for any~$P\in \mathcal{A}$, the matrix
$$
\int_{SO(3)} \big(\partial_t (\rho M_\Lambda)+A \vezero \cdot \nabla_x(\rho M_\Lambda) \big)\, \bar \psi_0(\Lambda \cdot A)\, \Lambda^TA\, dA =0.
$$
is orthogonal to all antisymmetric matrices. Therefore, it must be a symmetric matrix, meaning that we have
\be \label{eq:expression_to_compute}
X:=\int_{SO(3)} \big(\partial_t (\rho M_\Lambda)+A \vezero \cdot \nabla_x(\rho M_\Lambda) \big)\, \bar\psi_0(\Lambda \cdot A)\,(\Lambda^T A-A^T\Lambda) \, dA=0.
\ee

\medskip
We have with the definition of~$M_\Lambda$ in  \eqref{eq:Von_Mises_equilibria} that
\begin{eqnarray*}
\partial_t (\rho M_\Lambda) &=& M_\Lambda(\partial_t \rho +d^{-1}\nu(\Lambda\cdot A) \rho(A\cdot \partial_t \Lambda)),\\
(A \vezero \cdot \nabla_x) (\rho M_\Lambda) &=& M_\Lambda  \lp A\vezero \cdot \nabla_x \rho + d^{-1} \nu(\Lambda\cdot A) \, \rho (A\cdot (A\vezero \cdot \nabla_x) \Lambda) \rp.
\end{eqnarray*}
Inserting the two previous expressions into \eqref{eq:expression_to_compute}, we compute separately each component of~$X$ defined by:
\begin{eqnarray*}
X_1&:=& \int_{SO(3)} \partial_t \rho  M_\Lambda \, \bar\psi_0(\Lambda \cdot A)\,(\Lambda^T A-A^T\Lambda)\,  dA,\\
X_2 &:=& \int_{SO(3)} d^{-1}\nu(\Lambda\cdot A) \rho (A \cdot \partial_t \Lambda)  M_\Lambda \, \bar\psi_0(\Lambda \cdot A)\,(\Lambda^T A-A^T\Lambda) \, dA,\\
X_3 &:=& \int_{SO(3)} A\vezero \cdot \nabla_x \rho\,   M_\Lambda \, \bar\psi_0(\Lambda \cdot A)\,(\Lambda^T A-A^T\Lambda) \, dA,\\
X_4 &:=& \int_{SO(3)} d^{-1} \nu(\Lambda\cdot A) \, \rho (A\cdot (A\vezero \cdot \nabla_x) \Lambda) M_\Lambda\, \bar\psi_0(\Lambda \cdot A)\, (\Lambda^T A-A^T\Lambda) \, dA,
\end{eqnarray*}
so~$X= X_1 +X_2 +X_3+X_4$. 

\medskip 
For the first term we have (changing variables~$B=\Lambda^T A$):
\be \nn
X_1 =\partial_t \rho \int_{SO(3)}   M_{\Id}(B)\,\bar\psi_0(\Id\cdot B)\, (B-B^T) \, dB =0
\ee
since both~$M_{\Id}(B)$ and~$\bar \psi_0(\Id\cdot B)$ are invariant by the change~$B\mapsto B^T$.




For the  term~$X_2$ we make the change of variables~$B=\Lambda^T A$ and compute
\beqar
X_2 &=&  \rho\int_{SO(3)} d^{-1}\nu(\Id\cdot B) (\Lambda B \cdot \partial_t \Lambda)  M_{\Id}(B)\bar\psi_0(\Id \cdot B) (B-B^T) \, dB\\
&=&\frac{2d^{-1}\rho}{\pi Z}\int_{(0,\pi)\times S^2}\lp \Lambda \big( \Id + \sin\theta \asymn + (1-\cos\theta) [\nvec]_\times^2\big) \rp \cdot \partial_t \Lambda\\ 
&&\hspace{4cm}\sin^2(\theta/2)\,\nu(\tfrac12+\cos\theta) m(\theta)\,\widetilde\psi_0(\theta)\,2\sin\theta\,\asymn\, d\theta d\nvec,
\eeqar
where we have used the expression of the Haar measure~$dB = \frac{2}{\pi}\sin^2(\theta/2)d\theta d\nvec$ (see Lemma~\ref{lem-volume-form}) and that writing~$B=B(\theta, \nvec)=\Id + \sin\theta \asymn + (1-\cos\theta) [\nvec]_\times^2$ thanks to Rodrigues' formula~\eqref{eq:Rodrigues_formula}, we have~$B-B^T = 2\sin\theta \asymn$. Removing odd terms with respect to the change~$\nvec \mapsto -\nvec$, we obtain 
\[
X_2= \frac{4d^{-1}\rho}{\pi Z} \int_{(0,\pi)\times S^2} \,\nu(\tfrac12+\cos\theta) \sin^2\theta \,m(\theta)\widetilde\psi_0(\theta)\sin^2(\theta/2) (\Lambda \asymn \cdot \partial_t \Lambda)\, \asymn \, d\theta d\nvec.\]
Now since~$\partial_t\Lambda\in T_\Lambda$, we have~$\Lambda^T \partial_t\Lambda\in \mathcal{A}$ (antisymmetric, see Prop.~\ref{prop:tangentspaceSOn}), and so
$$\Lambda^T \partial_t\Lambda= [\pmb{\lambda_t}]_\times$$
for some vector~$\pmb{\lambda_t}$.
Therefore
$$(\Lambda \asymn)\cdot \partial_t \Lambda= \asymn \cdot (\Lambda^T \partial_t \Lambda) = \asymn \cdot [\pmb{\lambda}_t]_\times= (\nvec \cdot \pmb{\lambda_t}).$$
So using the definition~\eqref{eq:weight_mtilde} of~$\widetilde{m}(\theta)$, we get
\beqar
X_2&=& \frac{4d^{-1}\rho}{\pi Z}\int_{(0,\pi)\times S^2} \widetilde{m}(\theta)\sin^2(\theta/2) (\nvec\cdot \pmb{\lambda_t})\, \asymn \, d\theta d\nvec\\
&=& \frac{4d^{-1}\rho}{\pi Z} \left[\int_{(0,\pi)\times S^2} \widetilde{m}(\theta)\sin^2(\theta/2) (\nvec\cdot \pmb{\lambda_t})\, \nvec \, d\theta d\nvec \right]_\times\\
&=& \frac{4d^{-1}\rho}{3\pi Z} \lp \int^\pi_0  \widetilde{m}(\theta)\sin^2(\theta/2)\, d\theta \rp [\pmb{\lambda_t}]_\times,
\eeqar
because the mapping~$\ww \mapsto [\ww]_\times$ is linear, and~$\int_{S^2}\nvec\otimes \nvec \, d\nvec = \frac13\Id$.

Denote by
\[C_2 :=\frac{4d^{-1}}{3\pi Z} \lp \int^\pi_0  \widetilde{m}(\theta)\sin^2(\theta/2)\, d\theta \rp,\]
then we conclude that
$$X_2= C_2\rho\Lambda^T\partial_t \Lambda.$$
\medskip

Now, for the  term~$X_3$ we compute the following, starting again by the change of variables~$B=\Lambda^TA$:
\beqar
X_3 &=& \int_{SO(3)} (\Lambda B \vezero\cdot \nabla_x \rho) M_{\Id}(B) \,\bar\psi_0(Id \cdot B)\, (B-B^T) \, dB\\
&=& \frac{4}{\pi Z}\int_{(0,\pi)\times S^2}m(\theta)\,\widetilde\psi_0(\theta)\,\sin\theta\,\sin^2(\theta/2) \\
&&\hspace{3cm}\lp\Lambda\big(\Id +\sin\theta\asymn+(1-\cos\theta)[\nvec]_\times^2\big) \vezero\cdot \nabla_x\rho \rp\, \asymn\, d\theta d\nvec\\
&=& \frac{4}{\pi Z}\int_{(0,\pi)\times S^2}m(\theta)\,\widetilde\psi_0(\theta)\,\sin^2\theta\,\sin^2(\theta/2)\lp\Lambda \asymn \vezero\cdot \nabla_x\rho \rp\, \asymn\, d\theta d\nvec\\
&=& \frac{4}{\pi Z}\left[\int_{(0,\pi)\times S^2}\frac{\widetilde{m}(\theta)}{\nu(\frac12+\cos\theta)}\,\sin^2(\theta/2)\big( \nvec \cdot (\vezero\times \Lambda^T\nabla_x\rho)\big) \, \nvec d\theta d\nvec\right]_\times\\
&=&\frac{4}{3\pi Z}\lp \int^\pi_0 \frac{\widetilde{m}(\theta)}{\nu(\frac12+\cos\theta)}\sin^2(\theta/2)\, d\theta \rp [\vezero \times \Lambda^T \nabla_x \rho]_\times,
\eeqar
where we used similar considerations as for~$X_2$, as well as that
$$\Lambda \asymn \vezero\cdot \nabla_x \rho  =\asymn\vezero \cdot (\Lambda^T\nabla_x\rho) = (\nvec\times \vezero) \cdot (\Lambda^T\nabla_x\rho) = \nvec \cdot (\vezero\times \Lambda^T \nabla_x\rho).$$
Denote by
\[C_3 :=\frac{4}{3\pi Z}\lp \int^\pi_0 \frac{\widetilde{m}(\theta)}{\nu(\frac12+\cos\theta)}\sin^2(\theta/2)\, d\theta\rp,\]
then 
$$X_3=C_3[\vezero \times \Lambda^T\nabla_x\rho]_x.$$

\medskip

We now compute~$X_4$ in the same way, with the change of variables~$B=\Lambda^T A$:
\[X_4 = \rho d^{-1}\int_{SO(3)} \big( \nu(\Id\cdot B)  (\Lambda B\cdot (\Lambda B\vezero\cdot \nabla_x)\Lambda)\big) M_{\Id}(B)\,\bar\psi_0(\Id\cdot B)\, (B-B^T) \, dB\,.\]
We now use the definition of~$\mathcal{D}_x(\Lambda)$ given in~\eqref{eq-def-D} to get
\[X_4=\rho d^{-1} \int_{SO(3)} \big( \nu(\Id\cdot B) (\Lambda B\cdot ([\mathcal{D}_x(\Lambda)\Lambda B\vezero]_\times\Lambda)\big) M_{\Id}(B) (B-B^T) \bar\psi_0(\Id\cdot B)\, dB\,.\]

Using the fact that~$\Lambda^T[\ww]_\times=[\Lambda^T\ww]_\times\Lambda^T$ for all~$\ww\in\R^3$, we have
\[\Lambda B\cdot ([\mathcal{D}_x(\Lambda)\Lambda B\vezero]_\times\Lambda)=B\cdot [\Lambda^T\mathcal{D}_x(\Lambda)\Lambda B\vezero]_\times.\]
To simplify the notations, we denote~$L=\Lambda^T\mathcal{D}_x(\Lambda)\Lambda$. 
Since the symmetric part of~$B$ does not contribute to the scalar product $B\cdot[LB\vezero]_\times$, we get
\[\Lambda B\cdot ([\mathcal{D}_x(\Lambda)\Lambda B\vezero]_\times\Lambda)=B \cdot [LB \vezero]_\times= \sin\theta\,\asymn \cdot [LB\vezero]_\times = \sin\theta\, \nvec \cdot LB\vezero,\]
Therefore we obtain, in the same manner as before,
\[X_4=\frac{4\rho d^{-1}}{\pi Z}\int^\pi_0 \widetilde{m}(\theta)\sin^2(\theta/2) \left[ \int_{S^2} \Big(\nvec \cdot\big( L( \Id + \sin\theta\asymn+(1-\cos\theta)[\nvec]_\times^2)\vezero\big)\Big)\, \nvec \, d\nvec\right]_\times d\theta,
\]
and we have to know the value of
\beqar
\mathbf{y}(\theta)&:=& \int_{S^2} \Big(\nvec \cdot\big( L( \Id + \sin\theta\asymn+(1-\cos\theta)[\nvec]_\times^2)\vezero\big)\Big)\, \nvec \, d\nvec\\
&=&\int_{S^2} \bigg(\nvec \cdot\Big( L\lp \cos\theta \vezero + (1-\cos\theta)(\nvec\cdot\vezero)\, \nvec \rp \Big)\bigg)\nvec\, d\nvec\\
&=&\frac{1}{3}\cos\theta L\vezero+ (1-\cos\theta)\lp\int_{S^2} \nvec \cdot L\nvec \, (\nvec \otimes \nvec) \, d\nvec \rp \vezero,
\eeqar
where the term involving~$\asymn$ vanishes since its integrand is odd with respect to~$\nvec\mapsto-\nvec$.

To compute the second term of this expression we will make use of the following lemma proved at the end of this section:
\begin{lemma}
\label{lem:auxiliary_compute_X4} For a given matrix~$L\in \mathcal{M}$, we have
	$$\int_{S^2} \nvec \cdot L\nvec \, (\nvec \otimes \nvec)\, d\nvec = \frac{1}{15}(L+L^T)+ \frac{1}{15}\tr(L)\Id.$$
\end{lemma}
Using this lemma we have that
\beqar
\mathbf{y}(\theta) &=& \tfrac{1}{3}\cos\theta \,L\vezero + (1-\cos\theta)\big( \tfrac{1}{15}(L+L^T) + \tfrac{1}{15}\tr(L) \Id \big)\vezero\\
&=&\tfrac{1}{15}(1+4\cos\theta) L \vezero  + \tfrac{1}{15}(1-\cos\theta)\big( L^T\vezero+\tr(L)\vezero\big).
\eeqar
Therefore we obtain
\beqar
X_4 &=&\frac{4\rho d^{-1}}{\pi Z}\int^\pi_0 \widetilde{m}(\theta)\sin^2(\theta/2) [\mathbf{y}(\theta)]_\times \, d\theta\\
&=& \frac{4\rho d^{-1}}{15\pi Z}\int^\pi_0 \widetilde{m}(\theta)\sin^2(\theta/2) \big((1+4\cos\theta)[L\vezero]_\times +(1-\cos\theta)[ L^T\vezero+\tr(L)\vezero]_\times\big) d\theta\\
&=& \rho \big( C_4  [L\vezero]_\times+ C_5[ L^T\vezero+\tr(L)\vezero]_\times \big)
\eeqar
for 
\beqar
 C_4 &:=& \frac{4 d^{-1}}{15\pi Z}\int^\pi_0 \widetilde{m}(\theta)\sin^2(\theta/2) (1+4\cos\theta) \, d\theta,\\
C_5 &:=& \frac{4 d^{-1}}{15\pi Z}\int^\pi_0 \widetilde{m}(\theta)\sin^2(\theta/2) (1-\cos\theta) \, d\theta.
\eeqar
\bigskip

Finally putting all the terms together we have that
\beqar
0= X&=& X_1+X_2+X_3+X_4\\
&=&  C_2 \rho \Lambda^T\partial_t \Lambda+C_3 [\vezero \times \Lambda^T \nabla_x\rho]_\times  +\rho C_4 [L\vezero]_\times + \rho C_5[ L^T\vezero+\tr(L)\vezero]_\times.
\eeqar
In particular~$\Lambda X = 0$ and from the fact that~$\Lambda[\ww]_\times = [\Lambda\ww]_\times \Lambda$ for any~$\ww\in\R^3$ we get
\begin{equation}\label{eqL1}
0=\Lambda X = C_2 \rho \partial_t\Lambda +C_3 [(\Lambda\vezero) \times  \nabla_x\rho]_\times \Lambda +C_4\rho [\Lambda L\vezero ]_\times \Lambda + C_5\rho[ \Lambda L^T\vezero+\tr(L)\Lambda\vezero]_\times \Lambda.
\end{equation}
Since we have taken~$L=\Lambda^T\mathcal{D}_x(\Lambda)\Lambda$, we get that~$\tr(L)=\tr\big(\mathcal{D}_x(\Lambda)\big)=\delta_x(\Lambda)$ and, thanks to~\eqref{def-delta-r}:
\[[\Lambda L^T\vezero]_\times=[\mathcal{D}_x(\Lambda)^T\Lambda\vezero]_\times=[(\mathcal{D}_x(\Lambda)-[\rvec_x(\Lambda)]_\times)\Lambda\vezero]_\times\]
Furthermore, we have~$[\Lambda L\vezero]_\times\Lambda=[\mathcal{D}_x(\Lambda)\Lambda\vezero]_\times\Lambda=\big((\Lambda \vezero) \cdot\nabla_x\big)\Lambda$ thanks to the definition of~$\mathcal{D}_x$ given in~\eqref{eq-def-D}.
Finally, inserting these expressions into~\eqref{eqL1} and dividing by~$C_2$, we get the equation
\[
\rho\Big( \partial_t\Lambda + c_2\big((\Lambda \vezero) \cdot\nabla_x\big)\Lambda \Big)+ c_3[(\Lambda \vezero) \times \nabla_x\rho]_\times \Lambda + c_4\rho[-\rvec_x(\Lambda)\times (\Lambda\vezero) + \delta_x(\Lambda) \, \Lambda \vezero]_\times \Lambda=0,
\]
for
\beqar
c_2&=& \frac{C_4+C_5}{C_2}=\tfrac15\langle 2+3\cos\theta\rangle_{\widetilde{m}(\theta)\sin^2(\theta/2)}, \\
 c_3 &=& \frac{C_3}{C_2}=d\langle \nu(\tfrac12+\cos\theta)^{-1}\rangle_{\widetilde{m}(\theta)\sin^2(\theta/2)}, \\
 c_4&=& \frac{C_5}{C_2}=\tfrac15\langle 1-\cos\theta\rangle_{\widetilde{m}(\theta)\sin^2(\theta/2)},
\eeqar
which ends the proof.
\end{proof}      
    
\begin{proof}[Proof of Lemma~\ref{lem:auxiliary_compute_X4}]
Denote by~$\mathcal{I}(L)$ the integral that we want to compute
$$\mathcal{I}(L):= \int_{S^2}\nvec \cdot L\nvec\, (\nvec\otimes \nvec)\, d\nvec,$$
then, written in components, we have
\beqar
\mathcal{I}(L)_{ij} &=& \int_{S^2}(\nvec \cdot L\nvec) \, (\mathbf{e_i} \cdot \nvec) \, (\mathbf{e_j}\cdot \nvec) \, d\nvec\\
&=&  \left\{\begin{array}{ll}
(L_{ij}+L_{ji}) \int_{S^2}(\mathbf{e_i}\cdot \nvec)^2(\mathbf{e_j}\cdot \nvec)^2 \, d\nvec & \mbox{ if } i\neq j\\
\sum_{k} L_{kk}\int_{S^2}(\mathbf{e_k}\cdot \nvec)^2(\mathbf{e_i}\cdot \nvec)^2 \, d\nvec & \mbox{ if } i=j\\ 
\end{array} \right. \\
&=&\left\{\begin{array}{ll}
\frac{1}{15}(L_{ij}+L_{ji})  & \mbox{ if } i\neq j\\
\frac{1}{15}\sum_{k} L_{kk} + \frac{2}{15} L_{ii} & \mbox{ if } i=j\\ 
\end{array} \right. \\
&=&\frac{1}{15}\lp L_{ij}+L_{ji}\rp +\left\{\begin{array}{ll}
0  & \mbox{ if } i\neq j\\
\frac{1}{15}\sum_{k} L_{kk}  & \mbox{ if } i=j\\ 
\end{array} \right. \,,
\eeqar
from which we conclude the lemma.
In the computations we used that 
\beqar
\text{for }i\neq j, && \int_{S^2}(\mathbf{e_i}\cdot \nvec)^2(\mathbf{e_j}\cdot \nvec)^2 \, d\nvec =\frac{1}{4\pi} \int_{[0,\pi]\times[0,2\pi]}\sin^3 \phi\cos^2\psi\cos^2\phi \, d\phi d\psi =\frac{1}{15};\\
\text{for }k=i, && \int_{S^2}(\mathbf{e_k}\cdot \nvec)^4 \, d\nvec = \frac{1}{4\pi}\int_{[0,\pi]\times[0,2\pi]}\cos^4\phi\sin\phi\, d\phi d\psi = \frac{1}{5}.
\eeqar

\end{proof}     

Finally, we consider the orthonormal basis given by 
\[ 
\{\Lambda \vezero =:\Omega,\, \Lambda \mathbf{e_2}=:\uu, \, \Lambda\mathbf{e_3}=:\vv \},\]
where~$\{\mathbf{e_1}, \, \mathbf{e_2},\, \mathbf{e_3}\}$ is the canonical basis of~$\R^3$. We can have an expression of the operators~$\delta_x$ and~$\rvec_x$ in terms of these unit vectors~$\{\Omega,\uu,\vv\}$, which allows to rewrite the evolution equation of~$\Lambda$ as three evolution equations for these vectors. 

\begin{proposition} We have
\label{prop:evol-omega-u-v}
\begin{align}
\delta_x(\Lambda) &= [(\Omega \cdot \nabla_x) \uu] \cdot \vv + [(\uu\cdot \nabla_x)\vv]\cdot \Omega + [(\vv\cdot \nabla_x)\Omega]\cdot \uu, \label{eq-delta-omega-u-v}\\
\rvec_x(\Lambda)&= (\nabla_x \cdot \Omega)\Omega + (\nabla_x \cdot\uu) \uu + (\nabla_x \cdot \vv) \vv.\label{eq-r-omega-u-v}
\end{align}

Consequently, we have the following evolution equations for~$\Omega$,~$\uu$, and~$\vv$, corresponding to the evolution equation of~$\Lambda$ given in~\eqref{eq:macro_lambda}:
\beqarl \label{eq:Omega}
&&\rho D_t\Omega + P_{\Omega^\perp}\Big( c_3 \nabla_x\rho +c_4 \rho \big((\nabla_x\cdot \uu)\,\uu +  (\nabla_x\cdot \vv)\,\vv\big) \Big) =0,\\ 
&& \rho D_t \uu - \lp c_3\, \uu \cdot \nabla_x \rho + c_4 \rho \nabla_x \cdot \uu\rp \Omega + c_4\rho \,\delta_x(\Omega,\uu,\vv) \,\vv =0, \label{eq:u}\\
&&\rho D_t \vv -  \lp c_3\, \vv \cdot \nabla_x \rho + c_4 \rho \nabla_x \cdot \vv\rp \Omega - c_4 \rho \,\delta_x(\Omega,\uu,\vv) \,\uu =0, \label{eq:v}
\eeqarl
where~$D_t := \partial_t + c_2 (\Omega\cdot \nabla_x)$, and where~$\delta_x(\Omega,\uu,\vv)$ is the expression of~$\delta_x(\Lambda)$ given by~\eqref{eq-delta-omega-u-v}.

\end{proposition}

\begin{proof}
We first prove~\eqref{eq-delta-omega-u-v}. We have
\beqar
\delta_x(\Lambda) &=&\tr(\mathcal{D}_x(\Lambda))= \tr(\Lambda^T\mathcal{D}_x(\Lambda)\Lambda)=\sum_{k} \Lambda^T\mathcal{D}_x(\Lambda)\Lambda\mathbf{e_k}\cdot \mathbf{e_k} =\sum_{k} \big(\mathcal{D}_x(\Lambda)\Lambda\mathbf{e_k}\big)\cdot \Lambda\mathbf{e_k}\\
&=& \sum_{k} [\mathcal{D}_x(\Lambda)\Lambda\mathbf{e_k}]_\times \cdot [\Lambda \mathbf{e_k}]_\times=\sum_{k}[\mathcal{D}_x(\Lambda)\Lambda\mathbf{e_k}]_\times\Lambda  \cdot [\Lambda \mathbf{e_k}]_\times\Lambda\\
&=&\sum_{k} \lp(\Lambda \mathbf{e_k}\cdot \nabla_x)\Lambda \rp \cdot [\Lambda \mathbf{e_k}]_\times \Lambda,
\eeqar
thanks to the definition of~$\mathcal{D}_x$ given in \eqref{eq-def-D}. Now we use the fact that for two matrices~$A$,~$B$, we have~$A\cdot B=\frac12\tr(A^TB)=\frac12\sum_{i}A\mathbf{e_i}\cdot B\mathbf{e_i}$ (half the sum of the scalar products of the corresponding columns of the matrices~$A$ and~$B$), to get
\beqar
\delta_x(\Lambda)&=&\frac{1}{2}\sum_{k}\sum_{i} \big[\lp \Lambda \mathbf{e_k}\cdot \nabla_x\rp (\Lambda \mathbf{e_i})\big]\, \cdot \, \big[(\Lambda \mathbf{e_k})\times (\Lambda \mathbf{e_i})\big]\\
&=&\frac{1}{2}\Big((\Omega \cdot \nabla_x)\uu \cdot \vv - (\uu\cdot \nabla_x)\Omega \cdot \vv  - (\Omega\cdot \nabla_x) \vv\cdot \uu \\
&&\hspace{1cm}+ (\vv\cdot \nabla_x)\Omega\cdot \uu + (\uu\cdot \nabla_x)\vv\cdot \Omega - (\vv\cdot\nabla_x)\uu\cdot \Omega \Big)\\
&=& [(\Omega \cdot \nabla_x) \uu]\cdot \vv +[(\uu\cdot \nabla_x)\vv]\cdot \Omega +[(\vv\cdot \nabla_x)\Omega]\cdot \uu\, .
\eeqar
 For this last equality we used the fact that 
\[
0=(\Omega \cdot \nabla_x)(\uu\cdot \vv) =(\Omega \cdot \nabla_x)\uu\cdot \vv + (\Omega \cdot \nabla_x) \vv \cdot \uu\]
since~$\uu \perp \vv$ and analogously for the other components.

\medskip
We proceed next to proving the expression of $\rvec_x(\Lambda)$ given by~\eqref{eq-r-omega-u-v}. We first prove that~$\rvec_x(\Lambda)\cdot\Omega=\nabla_x\cdot\Omega$. We have (recall that~$[\rvec_x(\Lambda)]_\times=\mathcal{D}_x(\Lambda)-\mathcal{D}_x(\Lambda)^T$ and that for all~$\ww$ in~$\R^3$,~$\ww\cdot\nabla_x\Lambda=[\mathcal{D}_x(\Lambda)\ww]_\times\Lambda$):
\begin{align*}
\rvec_x(\Lambda)\cdot\Omega&=\rvec_x(\Lambda)\cdot(\uu\times\vv)=\vv\cdot([\rvec_x(\Lambda)]_\times\uu)=\vv\cdot\big(\mathcal{D}_x(\Lambda)-\mathcal{D}_x(\Lambda)^T\big)\uu\\
&=\vv\cdot\mathcal{D}_x(\Lambda)\uu-\uu\cdot\mathcal{D}_x(\Lambda)\vv\\
&=(\Omega\times\uu)\cdot\mathcal{D}_x(\Lambda)\uu+(\Omega\times\vv)\cdot\mathcal{D}_x(\Lambda)\uu\\
&=[\mathcal{D}_x(\Lambda)\uu]_\times\Omega\cdot\uu+[\mathcal{D}_x(\Lambda)\vv]_\times\Omega\cdot\vv\\
&=[\mathcal{D}_x(\Lambda)\uu]_\times\Lambda\mathbf{e_1}\cdot\uu+[\mathcal{D}_x(\Lambda)\vv]_\times\Lambda\mathbf{e_1}\cdot\vv\\
&=\big((\uu\cdot\nabla_x)\Lambda\mathbf{e_1}\big)\cdot\uu+\big((\vv\cdot\nabla_x)\Lambda\mathbf{e_1}\big)\cdot\vv\\
&=\big((\uu\cdot\nabla_x)\Omega\big)\cdot\uu+\big((\vv\cdot\nabla_x)\Omega\big)\cdot\vv.
\end{align*}
Since~$(\Omega\cdot\nabla_x)\Omega$ is orthogonal to~$\Omega$, we therefore get
\begin{align*}
\rvec_x(\Lambda)\cdot\Omega&=\big((\Omega\cdot\nabla_x)\Omega\big)\cdot\Omega+\big((\uu\cdot\nabla_x)\Omega\big)\cdot\uu+\big((\vv\cdot\nabla_x)\Omega\big)\cdot\vv\\
&= \sum_{i,k,j} \Lambda_{ik}\partial_i \Omega_j \Lambda_{jk} =\sum_{i,j}\partial_i \Omega_j\sum_{k} \Lambda_{ik}\Lambda_{kj}^T=\sum_{i}\partial_i \Omega_i=\nabla_x\cdot\Omega,\nonumber
\end{align*}
since~$\Lambda\Lambda^T=\Id$ (the first line is actually the expression of the divergence of~$\Omega$ in the basis~$\{\Omega, \uu, \vv\}$).
For the other two components of~$\rvec_x(\Lambda)$, we perform exactly the same computations with a circular permutation of the roles of~$\Omega,\uu,\vv$ to get~$\rvec_x(\Lambda) \cdot \uu = \nabla_x \cdot \uu$ and~$\rvec_x(\Lambda)\cdot \vv = \nabla_x\cdot \vv$. Therefore we obtain~\eqref{eq-r-omega-u-v}.

Finally we rewrite the equation for~$\Lambda$ as the evolution of the basis~$\{\Omega,\uu, \vv\}$. To obtain the evolution of~$\Lambda \mathbf{e_k}$ for~$k=1,2,3$, we multiply the Eq. \eqref{eq:macro_lambda} by~$\mathbf{e_k}$ and compute to obtain:
\beqar
&&\rho D_t\Omega + P_{\Omega^\perp}\lp c_3 \nabla_x\rho + c_4 \rho \,\rvec_x(\Lambda) \rp =0,\\ 
&& \rho D_t \uu - \uu \cdot \lp c_3 \nabla_x \rho + c_4 \rho \,\rvec_x(\Lambda)\rp \Omega + c_4\rho \delta_x(\Lambda)\, \vv =0,\\
&&\rho D_t \vv - \vv \cdot \lp c_3 \nabla_x \rho + c_4 \rho \,\rvec_x(\Lambda)\rp \Omega - c_4 \rho \delta_x(\Lambda)\, \uu =0,
\eeqar
where~$D_t= \partial_t + c_2(\Omega\cdot\nabla_x)$. To perform the computations we have used, for~$\ww=\nabla_x\rho$ or~$\ww=\rvec$ that
$$[\ww\times\Omega]_\times \Omega = -P_{\Omega^\perp}(\ww)\quad\text{ and }\quad(\ww\times \Omega)\times \uu = (\uu\cdot \ww)\Omega$$
since~$\Omega\perp \uu$ (analogously for~$\vv$). From here, using~\eqref{eq-r-omega-u-v} we obtain straightforwardly Eqs.~\eqref{eq:Omega},~\eqref{eq:u}, and~\eqref{eq:v} for~$\Omega$,~$\uu$ and~$\vv$  respectively.
\end{proof}

\section{Conclusions and open questions}

In the present work we have presented a new flocking model through body attitude coordination. We have proposed an  Individual Based Model where  agents are described by their position and a rotation matrix (corresponding to the body attitude). From the Individual Based Model we have derived the macroscopic equations via the mean-field equations. We observe that the macroscopic equation gives rise to  a new class of models, the Self-Organized Hydrodynamics for body attitude coordination (SOHB).
This model does not reduce to the more classical Self-Organized Hydrodynamics (SOH), which is the continuum version of the Vicsek model.
The dynamics of the SOHB system are more complex than those of the SOH ones of the Vicsek model. In a future work, we will carry out simulations of the Individual Based Model and the SOHB model  and study the patterns that arise to compare them with the ones of the Vicsek and SOH model. 

Also, there exist yet many open questions on the modelling side. For instance, one could consider that agents have a limited angle of vision, thus the so-called influence kernel~$K$ (see Section~\ref{sec:derivation_IBM}) is not isotropic any more, see \cite{frouvelle2012continuum} for the case of the Vicsek and SOH models. One could also consider a different interaction range for the influence kernel~$K$ that may give rise to  a diffusive term in the macroscopic equations, see \cite{degond2013macroscopic}. Moreover, in the case of the SOH model, when the coordination frequency and noise intensity (quantities~$\nu$ and~$D$ in the Individual Based Model \eqref{eq:IBM}-\eqref{eq:IBM2}) are functions of the flux of the agents, then phase transitions occur at the macroscopic level \cite{degond2013macroscopic}, (see also \cite{barbaro2012phase,bertin2009hydrodynamic,degond2015phase,toner1995long}). An analogous feature is expected to happen in the present case. Finally, one could think of elaborating on the model by adding repulsive effects at short range and attraction effects at large range.

On the analytical side, this model opens also many questions like making Prop.~\ref{prop:mean_field_limit}  rigorous, which means dealing with Stochastic Differential Equations with non-Lipschitz coefficients. In the context of the Vicsek model, the global well-posedness has been proven for the homogeneous mean-field Vicsek equation and also its convergence to the von Mises equilibria in \cite{Figalli}, see also \cite{Gamba}; an analogous result for our model will be desirable. The convergence of the Vicsek model to the model which was formally done in  \cite{degond2008continuum} has been recently achieved rigorously in \cite{Hydro_limit}. Again, one could also think of generalizing these results to our case.

     \bigskip
     
     \paragraph{Acknowledgements.}
     
P.D. acknowledges support from the Royal Society and the Wolfson foundation through a Royal Society Wolfson Research Merit Award; the National Science Foundation under NSF Grant RNMS11-07444 (KI-Net); the British ``Engineering and Physical Research Council''	 under grant ref: EP/M006883/1. 
P.D. is on leave from CNRS, Institut de Math\'ematiques de Toulouse, France.

A.F acknowledges support for the ANR projet ``KIBORD'', ref: ANR-13-BS01-0004 funded by the
French Ministry of Research.

 S.M.A. was supported by the British ``Engineering and Physical Research Council'' under grant ref: EP/M006883/1. S.M.A. gratefully acknowledges the hospitality of CEREMADE, Universit\'e Paris Dauphine, where part of this research was conducted.

\appendix

\section{Special Orthogonal Group~$SO(3)$}
\label{ap:SO(3)}

Throughout the text, we used repeatedly the following properties:
\begin{proposition}[Space decomposition in symmetric and antisymmetric matrices] \label{prop:spacedecompositionsymmetricandanti}
Denote by~$\mathcal{S}$ the set of symmetric matrices in~$\mathcal{M}$ and by~$\mathcal{A}$ the set of antisymmetric ones. Then
$$\mathcal{S} \oplus \mathcal{A}=\mathcal{M} \mbox{ and } \mathcal{A}\perp \mathcal{S}.$$
\end{proposition}
\begin{proof}For~$A\in\mathcal{M}$ we have~$A=\frac12(A+A^T)+\frac12(A-A^T)$, the first term being symmetric and the second antisymmetric. The orthogonality comes from the properties of the trace, namely~$\tr(A^T)=\tr(B)$, and~$\tr(AB)=\tr(BA)$~for~$B\in\mathcal{M}$. Indeed if~$P\in \mathcal{A}$ and~$S \in \mathcal{S}$ then~$\tr(P^TS)=\tr(SP^T)=\tr(PS^T)=-\tr(P^TS)$. Hence~$P \cdot S=\frac12\tr(P^TS)=0$.
\end{proof} 
\begin{proposition}[Tangent space to~$SO(3)$] \label{prop:tangentspaceSOn}
For~$A\in SO(3)$, denote by~$T_A$ the tangent space to~$SO(3)$ at~$A$. Then 
\[M\in T_A \mbox{ if and only if there exists } P\in \mathcal{A}\mbox{ s.t } M=AP,\]
or equivalently the same statement with~$M=PA$. Consequently, we have that
\[M\in T^\perp_A \mbox{ if and only if there exists } S\in \mathcal{S} \mbox{ s.t. } M=AS, \]
or equivalently the same statement with~$M=SA$.
 \end{proposition}
\begin{proof}
We have that~$M \in T_A$ if and only if there exists a curve~$\Lambda(t)$ from the neighborhood of~$0$ in~$\R$ to~$SO(3)$ such that~$\Lambda(0)=A$ and~$\Lambda'(0)=M$. We then have
\[\Id=\Lambda(t)\Lambda^T(t)=(A+tM+o(t))(A^T+tM^T+o(t))=\Id +t(A^T M +M^TA)+o(t).\]
So if~$M\in T_A$, we must have~$(A^T M +M^TA)=0$, that is to say that~$P=A^TM\in\mathcal{A}$. 

Conversely if~$M=AP$ with~$P\in\mathcal{A}$, the solution of the linear differential equation~$\Lambda'(t)=\Lambda(t)P$ with~$\Lambda(0)=A$ is given by~$\Lambda(t)=Ae^{tP}$ so it is a curve in~$SO(3)$. Indeed we have~$\Lambda(t)^T\Lambda=(e^{tP})^Te^{tP}=e^{tP^T}e^{tP}=e^{-tP}e^{tP}=\Id$. Since~$\Lambda'(0)=AP=M$, we get that~$M\in T_A$. The equivalent condition comes from the fact that if~$M=AP$, with~$P\in\mathcal{A}$, then~$M=APA^TA=\widetilde PA$ with~$\widetilde P\in\mathcal{A}$. Finally the last part is obtained thanks to Prop.~\ref{prop:spacedecompositionsymmetricandanti}) and the fact that the dot product is left (and right) invariant with respect to~$SO(3)$: if~$B,C\in\mathcal{M}$ and~$A\in SO(3)$, then~$AB\cdot AC=\frac12\tr(B^TA^TAC)=B\cdot C$.
 \end{proof}
 
 \begin{proposition}[Projection operator on the tangent space]
 \label{prop:projection_tangentspace}
 Let~$A\in SO(3)$ and~$M\in \mathcal{M}$ (set of square matrices). Let~$P_{T_A}$ be the orthogonal projection on~$T_A$ (tangent space at~$A$), then
\begin{equation}
\label{eq:definition_projection}
 P_{T_A}(M) = \frac{1}{2}\lp M-AM^T A\rp.
 \end{equation}
 Notice that then
$$P_{T^\perp_A}(M) = \frac{1}{2} \lp M +AM^T A \rp.$$
\end{proposition}
\begin{proof}It suffices to verify that the expression given for~$P_{T_A}(M)$ satisfies~$P_{T_A}(M)\in T_A$ and~$M-P_{T_A}(M)\in T_A^\perp$, that is to say~$A^TP_{T_A}M\in\mathcal{A}$ and~$A^T(M-P_{T_A}(M))\in\mathcal{S}$ thanks to Prop.~\ref{prop:tangentspaceSOn}.
We have indeed~$A^T\frac12(M-AM^TA)=\frac12(A^TM-M^TA)$ which is clearly antisymmetric, and~$A^T\frac12(M+AM^TA)=\frac12(A^TM+M^TA)$ which is symmetric.
\end{proof}

To compute the gradient in~$SO(3)$ of a function~$\psi:SO(3)\to \R$ we will consider~$A(\eps)$ a differentiable curve in~$SO(3)$ such that
$$A(0)=A, \qquad \left. \frac{d}{d\eps}A(\eps)\right|_{\eps=0}= \delta_A\in T_A$$
then~$\nabla_A\psi(A)$ is the element of~$T_A$ such that for any~$\delta_A\in T_A$, we have
$$\lim_{\eps \to 0} \frac{\psi(A(\eps))-\psi(A)}{\eps}= \nabla_A\psi(A) \cdot \delta_A.$$
In particular, one can check that
\be \label{eq:gradient_product}
\nabla_A (A \cdot M) = P_{T_A}(M), \qquad M\in \mathcal{M}.
\ee
We now show that the differential equation corresponding to following this gradient has trajectories supported on geodesics.
\begin{proposition}
\label{prop:relax_geodesic}
If~$B\in SO(3)$ and~$A_0\in SO(3)$, the trajectory of the solution of the differential equation~$\frac{dA}{dt}=\nu(A\cdot B)\, P_{T_A}B=\nu(A\cdot B)\,\nabla_A (A \cdot B)$ with~$A(0)=A_0$ (and with~$\nu$ smooth and positive) is supported on a geodesic from~$A_0$ to~$B$.

\end{proposition} 
\begin{proof}
Indeed, write~$B^TA_0=\exp(\theta_0 \asymn)$ thanks to Rodrigues' formula~\eqref{eq:exponential_form} with~$\asymn$ an antisymmetric matrix of unit norm and~$\theta_0\in[0,\pi]$. If we set~$A(t)=B\exp(\theta(t)\asymn)$ where~$\theta$ satisfies the equation~$\theta'=-\nu(\frac12+\cos\theta)\sin\theta$ with~$\theta(0)=\theta_0$, we get 
\[\frac{dA}{dt}=B\exp(\theta(t)\asymn)\theta'(t)\asymn=-\nu(\tfrac12+\cos\theta(t)) B\exp(\theta(t)\asymn)\sin\theta(t)\asymn.\]
Now, thanks to the expression~\eqref{eq:Rodrigues_formula}, we have 
\[\sin\theta\asymn=\tfrac12(\exp(\theta\asymn)-\exp(\theta\asymn)^T)=\tfrac12(B^TA-A^TB),\]
and~$A\cdot B=\Id\cdot AB^T=\frac12\tr(\exp(\theta\asymn))=\frac12+\cos\theta$ thanks to~\eqref{eq:trace_theta}.
Therefore we obtain
\[\frac{dA}{dt}=-\nu(A\cdot B) A\,\tfrac12(B^TA-A^TB)=\nu(A\cdot B) P_{T_A}B,\] 
thanks to~\eqref{eq:definition_projection} and we have~$A(0)=A_0$. Since~$\theta\in[0,\theta_0]\mapsto\exp(\theta \asymn)$ is a geodesic between~$\Id$ and~$B^TA_0$, then~$\theta\mapsto B\exp(\theta \asymn)$ is a geodesic between~$B$ and~$A$, and the solution~$A(t)$ is supported on this geodesic. It is also easy to see that, except in the case~$\theta_0=\pi$ or~$\theta_0=0$, for which the solution is constant, the function~$t\mapsto \theta(t)$ (solution of the one-dimensional differential equation~$\theta'=-\nu(\frac12+\cos\theta)\sin\theta$) is positive, decreasing, and converge exponentially fast to~$0$, with an asymptotic exponential rate~$\nu(\frac32)$. Therefore, as time goes to infinity, the trajectory covers the whole geodesic from~$A_0$ to~$B$ (excluded).
\end{proof}


We now turn to the proofs of the expressions of the gradient, the volume form and the divergence in~$SO(3)$ in the so-called Euler axis-angle coordinates, that were presented in section~\ref{sec:differential_calculus}. 

\begin{proof}[Proof of Prop.~\ref{prop:gradient_SO3}: expression of the gradient in~$SO(3)$.]

Consider a curve in~$SO(3)$ given by 
$$A(t)= \exp(\theta(t)\asymn(t))= \Id+\sin(\theta(t))\asymn(t)+(1-\cos(\theta(t)))[\nvec]_\times^2(t)$$
(following \eqref{eq:Rodrigues_formula}-\eqref{eq:def_operator_asym}) with~$A(0)=A$,~$\theta(0)=\theta$ and~$\asymn(t)=\left[\nvec(t) \right]_\times$,~$\nvec(0)=\nvec$. Define:
\beqar
\delta_A &=& A'(0) \in T_A,\\
\delta_\theta &=& \theta'(0) \in \R,\\
\delta_\nvec &=& \nvec'(0),\\
\delta_\asymn &=& \asymn'(0) = [\delta_\nvec ]_\times.
\eeqar
With these notations, for a function~$f=f(A(\theta,\nvec))$ it holds:
\be \label{eq:looked_for_nabla}
\nabla_A f \cdot \delta_A= \frac{\partial f}{\partial \theta} \delta_\theta + \nabla_\nvec f \cdot \delta_\nvec.
\ee
On the other hand, it holds true that
\beqarl
\delta_A &=& A\asymn \delta_\theta +  \sin\theta \delta_\asymn + (1- \cos\theta) \lp \asymn\delta_\asymn + \delta_\asymn \asymn \rp\nonumber\\
&=& A\asymn \delta_\theta +A A^T\Big(\sin\theta \delta_\asymn + (1- \cos\theta) \lp \asymn\delta_\asymn + \delta_\asymn \asymn \rp \Big)\nonumber\\
&=& A\asymn \delta_\theta +A \big(\Id-\sin\theta\,\asymn+(1-\cos\theta)[\nvec]_\times^2\big)\nonumber\\
&&\hspace{3cm}\Big(\sin\theta \delta_\asymn + (1- \cos\theta) \lp \asymn\delta_\asymn + \delta_\asymn \asymn \rp \Big)\nonumber\\
&=& A\asymn \delta_\theta +A \Big(\sin\theta \delta_\asymn + (1- \cos\theta) \lp \delta_\asymn \asymn-\asymn\delta_\asymn  \rp \Big)\nonumber\\
&=& A\asymn \delta_\theta + 2 \sin(\theta/2)A \lp\cos(\theta/2)\left[\delta_\nvec \right]_\times + \sin(\theta/2) \left[ \nvec \times \delta_\nvec \right]_\times \rp, \nonumber\\
&=&A\asymn \delta_\theta+L_\asymn(\delta_\asymn)\label{eq:decomposition_tangent_space},
\eeqarl
where the last line defines~$L_\asymn$.
In the first line, the term in~$\delta_\theta$ is obtained by differentiating the exponential form \eqref{eq:exponential_form} of~$A(t)$ assuming that~$\asymn(t)$ is constant. The term in~$\delta_\asymn$ is obtained by differentiating Rodrigues' formula \eqref{eq:Rodrigues_formula}.
To do the computation we have used Rodrigues' formula \eqref{eq:Rodrigues_formula} to express~$A^T$ and the facts that~$[\nvec]_\times^3=-\asymn$;~$\asymn\delta_\asymn \asymn=0$; and~$\delta_\asymn \asymn-\asymn\delta_\asymn = [\delta_\nvec \times \nvec]_\times$.

In particular notice that~$\{\asymn, [\delta_\nvec]_\times, [\nvec \times \delta_\nvec]_\times \}$ is an orthogonal basis of~$\mathcal{A}$ (antisymmetric matrices) from which we obtain a basis of~$T_A$ (by Prop.~\ref{prop:tangentspaceSOn}). So, we just need to compute the components of~$\nabla_A f$ in~$span\{A\asymn\}$ and~$span\{(A\asymn)^\perp\}$.

We will show that the component in~$span\{A\asymn\}$ is given by
\be \label{eq:projection1}
P_{A\asymn}\lp \nabla_A f \rp = \frac{\partial f}{\partial \theta} A\asymn
\ee
and the one on~$span\{(A\asymn)^\perp\}$ is
\be \label{eq:projection2}
P_{(A\asymn)^\perp}\lp \nabla_A f \rp =\frac{1}{2 \sin(\theta/2)} A \lp\cos(\theta/2) \left[ \nabla_\nvec f\right]_\times+ \sin(\theta/2) \left[\nvec \times \nabla_\nvec f \right]_\times \rp.
\ee
The sum of the two previous expressions gives \eqref{eq:gradient_SO3} ($\nabla_A f= P_{A\asymn}(\nabla_A f) + P_{(A\asymn)^\perp}(\nabla_A f)$).
The component \eqref{eq:projection1} is computed considering the case where~$\delta_\nvec=0$ in \eqref{eq:decomposition_tangent_space}-\eqref{eq:looked_for_nabla}, so that
$$\nabla_A f \cdot \delta_A = \nabla_A f \cdot A\asymn \delta_\theta = \frac{\partial f}{ \partial \theta} \delta_\theta.$$
Expression \eqref{eq:projection1} is obtained by noticing that~$(A\asymn) \cdot (A\asymn) =\asymn\cdot \asymn=\nvec\cdot\nvec=1$ (using \eqref{eq:properties_asym_matrix}).

To obtain the component \eqref{eq:projection2}, consider  the case~$\delta_\theta=0$ in \eqref{eq:decomposition_tangent_space} and \eqref{eq:looked_for_nabla} so that
\begin{equation}
\nabla_A f \cdot \delta_A = \nabla_A f \cdot L_\asymn(\delta_\asymn) = \nabla_\nvec f \cdot \delta_\nvec,\label{eq-nabladelta}
\end{equation}
where~$L_\asymn$ is given in \eqref{eq:decomposition_tangent_space}.

We have that
$$P_{(A\asymn)^\perp}\lp\nabla_A f \rp= A \left[ \uu\right]_\times \quad \mbox{ for some } \uu \perp \nvec.$$
The goal is to compute~$\uu$ as a function of~$\vv:= \nabla_\nvec f$. By~\eqref{eq-nabladelta} we have that
$$A\left[ \uu \right]_\times \cdot L_\asymn(\delta_\asymn)= \nabla_\nvec f \cdot \delta_\nvec.$$
This implies that
$$ 2 \sin(\theta/2) \left[ \uu\right]_\times \cdot \lp \cos(\theta/2) \left[ \delta_\nvec \right]_\times + \sin(\theta/2) \left[ \nvec \times \delta_\nvec \right]_\times\rp = \vv \cdot \delta_\nvec \quad \mbox{for all } \delta_\nvec \perp \nvec,$$
so (see \eqref{eq:properties_asym_matrix}) we get
$$2 \sin(\theta/2) \lp\cos(\theta/2) \uu + \sin(\theta/2) \uu\times \nvec \rp \cdot \delta_\nvec= \vv \cdot \delta_\nvec.$$
Since this is true for all~$\delta_\nvec$ orthogonal to~$\nvec$, we get
\[\vv=2 \sin(\theta/2) \lp\cos(\theta/2) \uu + \sin(\theta/2) \uu\times \nvec \rp.\]
From here can get the expression of~$\nvec \times\vv$ in terms of~$\uu$ and~$\nvec \times \uu$. After some computations we finally obtain that
$$\uu= \frac{1}{2 \sin(\theta/2)}\lp \cos(\theta/2) \vv + \sin(\theta/2) \nvec \times \vv\rp.$$

\end{proof}

\begin{proof}[Proof of the volume form, Lemma~\ref{lem-volume-form}]
We denote by~$g$ the metric of the Riemannian manifold~$SO(3)$ associated to the inner product
$$A\cdot B = \frac{1}{2}\tr(A^TB), \quad A, B\in SO(3).$$
The volume form is proportional to~$\sqrt{\det(g)}$ \cite{gallot1990riemannian}. We compute the volume form using spherical coordinates, i.e., we consider the coordinates~$(\theta, \phi, \psi)\in[0,\pi]\times[0,\pi]\times[0,2\pi]$. Given the Euler axis-angle coordinates~$(\theta,\nvec)$ we have that
$$\nvec = \lp \begin{array}{c}
\sin \phi \cos\psi \\
\sin \phi \sin \psi \\
\cos \phi
\end{array}
\rp.$$
For the spherical coordinate system, we consider the vector field~$\lp\frac{\partial}{\partial \theta}, \frac{\partial}{\partial \phi}, \frac{\partial}{\partial \psi}\rp$. Denoting
$$Y_1= \frac{\partial A}{\partial \theta}, \,\; Y_2 = \frac{\partial A}{\partial \phi},\,\; Y_3= \frac{\partial A}{\partial \psi} , \quad A\in SO(3),$$
we get that~$(Y_{i})_{i=1,2,3}\in T_{A}(SO(3))$ forms a basis of vectors fields at~$A$.

The metric~$g$ is defined as~$g_{ij}=g(Y_i, Y_j)= \frac{1}{2}\tr (Y_i^T Y_j)$,~$i,j=1,2,3$. We compute next each term. Firstly, we know that for a given~$\delta_A \in T_A$, there exists~$\delta_\theta, \delta_\psi, \delta_\phi\in \R$ such that
$$\delta_A = \frac{\partial A}{\partial \theta} \delta_\theta + \frac{\partial A}{\partial \phi}\delta_\phi + \frac{\partial A}{\partial \psi} \delta_{\psi}$$
and also for a given~$\delta_\nvec \in T_{\nvec}(S^2)$ (the tangent plane to the sphere at~$\nvec$), there exists~$\delta_\psi'$,~$\delta_\psi'$ such that
$$\delta_\nvec = \frac{\partial \nvec}{\partial\phi}\delta_\phi' +\frac{\partial \nvec}{\partial \psi} \delta_\psi'.$$

Now, following the computation given in \eqref{eq:decomposition_tangent_space} we have that, for~$\delta_\theta =1, \delta_\phi=0, \delta_\psi=0$ 
$$\frac{\partial A}{\partial \theta}= \delta_A = A[\nvec]_\times.$$
Now, if~$\delta_\theta=0, \delta_\phi=1, \delta_\psi=0$ then, using that~$\delta_\nvec = \frac{\partial \nvec}{\partial \phi}$ we have that
$$\frac{\partial A}{\partial \phi}=\delta_A= 2\sin(\theta/2) A\left[R_{\nvec,\theta/2}\lp\frac{\partial \nvec}{\partial \phi} \rp\right]_\times,$$
where
$$R_{\nvec,\theta/2}(\vv)=\cos(\theta/2) \vv + \sin (\theta/2) (\nvec \times \vv),$$
which corresponds to the rotation of the vector~$\vv$ around~$\nvec$ by an angle~$\theta/2$ (anticlockwise) as long as~$\vv\cdot \nvec =0$.
Analogously one can also deduce that
$$\frac{\partial A}{\partial \psi}= 2 \sin(\theta/2) \left[R_{n,\theta/2}\lp \frac{\partial \nvec}{\partial \psi}\rp\right]_\times.$$

From here, using that~$\|\frac{\partial{\nvec}}{\partial\phi}\|^2=1$ and~$\|\frac{\partial{\nvec}}{\partial\psi}\|^2=\sin^2\phi$, we conclude that
$$g=\lp
\begin{array}{ccc}
1 & 0 & 0\\
0 & 4\sin^2(\theta/2) & 0 \\
0 & 0 & 4 \sin^2(\theta/2)\sin^2\phi
\end{array}
\rp.
$$
Notice that to compute~$g\lp\frac{\partial A}{\partial \theta}, \frac{\partial A}{\partial \phi} \rp$ we use that~$R_{\nvec,\theta/2}\lp\frac{\partial\nvec}{\partial\phi} \rp \perp \nvec$.

Finally we have that
$$\sqrt{det(g)}= 4 \sin^2(\theta/2)\sin \phi$$
and therefore
\beqar
\int_{SO(3)}f(A)\, dA &=& \int_{[0,\pi]\times [0,\pi]\times[0,2\pi]} \widetilde{f}(\theta,\phi, \psi) 4 \sin^2(\theta/2) \sin \phi \, d\theta d\phi d\psi\\
&=& 4\int_{\theta\in[0,\pi]} \lp \int_{ [0,\pi]\times [0.2\pi]} \tilde f(\theta,\phi,\psi)\sin\phi d\phi d\psi\rp \sin^2(\theta/2) \, d\theta.
\eeqar
The term~$\sin \phi d\phi d\psi$ is the volume element in the sphere~$S^2$ so we have that
$$\int_{S^2} \hat{f}(\theta,\nvec) d\nvec = \int_{[0,\pi]\times[0,2\pi]} \tilde f(\theta,\phi,\psi)\sin\phi d\phi d\psi.$$
Therefore, the volume element corresponding to the Euler axis-angle coordinates is proportional to~$\sin^2(\theta/2)d\theta d\nvec$. Since the volume element is defined up to a constant, we choose the constant~$c$ such that
$$\int^{2\pi}_0 c \sin^2(\theta/2)\, d\theta =1,$$
i.e.,~$c= 2/\pi$. In conclusion, the volume element in the Euler axis-angle coordinates corresponds to 
$$\frac{2}{\pi}\sin^2(\theta/2)d\theta d\nvec.$$

\end{proof}
	
\begin{proof}[Proof of divergence formula, Prop.~\ref{prop:divergence_SO3}]
We compute the divergence by duality of the gradient, Prop.~\ref{prop:gradient_SO3}. Let~$f=f(A)$ be a function and consider
\beqar
&&\int_{SO(3)} \nabla_A \cdot B(A)\, f(A) \, dA \\
&=& - \int_{SO(3)} B(A) \cdot \nabla_A f(A) \, dA\\
&=& - \int_{(0,\pi)\times S^2} W(\theta)   b(\theta,\nvec)  \partial_\theta  f(\theta,\nvec) \, d\theta d\nvec\\
&&\, - \int_{(0,\pi)\times S^2} \frac{ W(\theta) }{2 \sin(\theta/2)} \vv(\theta,\nvec) \cdot \big( \cos(\theta/2)\nabla_\nvec f(\nvec, \theta) + \sin(\theta/2) \nvec\times\nabla_\nvec f(\nvec, \theta)\big)  d\theta d\nvec\\
&=&\int_{(0,\pi)\times S^2} \frac{f(\theta,\nvec)}{\sin^2(\theta/2)}\partial_\theta \lp\sin^2(\theta/2) b(\theta,\nvec) \rp W(\theta)\, d\theta d\nvec\\
&&\,+\int_{(0,\pi)\times S^2}  \frac{f(\theta,\nvec) }{2\sin(\theta/2)}\nabla_\nvec \cdot \big( \vv(\theta,\nvec) \cos(\theta/2) + \sin(\theta/2) (\vv(\theta,\nvec) \times \nvec)\big) W(\theta)d \theta d\nvec,
\eeqar
where~$W$ is given by \eqref{eq:volume_element},
from which we deduce the result.
\end{proof}

\bibliographystyle{plain}
\bibliography{Bibliography/biblio_body_attitude}

\begin{thebibliography}{10}

\bibitem{aldana2003phase}
M.~Aldana and C.~Huepe.
\newblock Phase transitions in self-driven many-particle systems and related
  non-equilibrium models: a network approach.
\newblock {\em J. Stat. Phys.}, 112(1-2):135--153, 2003.

\bibitem{armbruster2006model}
D.~Armbruster, P.~Degond, and C.~Ringhofer.
\newblock A model for the dynamics of large queuing networks and supply chains.
\newblock {\em SIAM J. Appl. Math.}, 66(3):896--920, 2006.

\bibitem{aw2002derivation}
A.~Aw, A.~Klar, M.~Rascle, and T.~Materne.
\newblock Derivation of continuum traffic flow models from microscopic
  follow-the-leader models.
\newblock {\em SIAM J. Appl. Math.}, 63(1):259--278, 2002.

\bibitem{barbaro2012phase}
A.~B.T. Barbaro and P.~Degond.
\newblock Phase transition and diffusion among socially interacting
  self-propelled agents.
\newblock {\em arXiv preprint arXiv:1207.1926}, 2012.

\bibitem{ben2000cooperative}
E.~Ben-Jacob, I.~Cohen, and H.~Levine.
\newblock Cooperative self-organization of microorganisms.
\newblock {\em Advances in Physics}, 49(4):395--554, 2000.

\bibitem{bertin2009hydrodynamic}
E.~Bertin, M.~Droz, and G.~Gr{\'e}goire.
\newblock Hydrodynamic equations for self-propelled particles: microscopic
  derivation and stability analysis.
\newblock {\em J. Phys. A: Math. Theor.}, 42(44):445001, 2009.

\bibitem{bolley2012mean}
F.~Bolley, J.~A. Ca{\~n}izo, and J.~A. Carrillo.
\newblock Mean-field limit for the stochastic {V}icsek model.
\newblock {\em Appl. Math. Lett.}, 25(3):339--343, 2012.

\bibitem{buhl2006disorder}
J.~Buhl, D.~Sumpter, I.~D. Couzin, J.~J. Hale, E.~Despland, E.~R. Miller, and
  S.~J. Simpson.
\newblock From disorder to order in marching locusts.
\newblock {\em Science}, 312(5778):1402--1406, 2006.

\bibitem{cavagna2010scale}
A.~Cavagna, A.~Cimarelli, I.~Giardina, G.~Parisi, R.~Santagati, F.~Stefanini,
  and M.~Viale.
\newblock Scale-free correlations in starling flocks.
\newblock {\em Proc. Natl. Acad. Sci.}, 107(26):11865--11870, 2010.

\bibitem{cavagna2014flocking}
A.~Cavagna, L.~Del~Castello, I.~Giardina, T.~Grigera, A.~Jelic, S.~Melillo,
  T.~Mora, L.~Parisi, E.~Silvestri, M.~Viale, et~al.
\newblock Flocking and turning: a new model for self-organized collective
  motion.
\newblock {\em J. Stat. Phys.}, 158(3):601--627, 2014.

\bibitem{cercignani2013mathematical}
C.~Cercignani, R.~Illner, and M.~Pulvirenti.
\newblock {\em The mathematical theory of dilute gases}, volume 106.
\newblock Springer Science \& Business Media, 2013.

\bibitem{colbois}
B.~Colbois.
\newblock Laplacian on {R}iemannian manifolds.
\newblock {\em Notes of a series of 4 lectures given in Carthage}, May 21-22,
  2010.

\bibitem{constantin2010onsager}
P.~Constantin.
\newblock The {O}nsager equation for corpora.
\newblock {\em J. Comput. Theor. Nanosci.}, 7(4):675--682, 2010.

\bibitem{constantin2010high}
P.~Constantin, A.~Zlato{\v{s}}, et~al.
\newblock On the high intensity limit of interacting corpora.
\newblock {\em Commun. Math. Sci.}, 8(1):173--186, 2010.

\bibitem{couzin2002collective}
I.~D. Couzin, J.~Krause, R.~James, G.~D. Ruxton, and N.~R. Franks.
\newblock Collective memory and spatial sorting in animal groups.
\newblock {\em J. Theoret. Biol.}, 218(1):1--11, 2002.

\bibitem{degond2004macroscopic}
P.~Degond.
\newblock Macroscopic limits of the {B}oltzmann equation: a review.
\newblock In {\em Modeling and Computational Methods for Kinetic Equations},
  pages 3--57. Springer, 2004.

\bibitem{degond2014hydrodynamics}
P.~Degond, G.~Dimarco, and T.~B.~N. Mac.
\newblock Hydrodynamics of the {K}uramoto--{V}icsek model of rotating
  self-propelled particles.
\newblock {\em Math. Models Methods Appl. Sci.}, 24(02):277--325, 2014.

\bibitem{degond2014macroscopic}
P.~Degond, G.~Dimarco, T.B.~N. Mac, and N.~Wang.
\newblock Macroscopic models of collective motion with repulsion.
\newblock {\em arXiv preprint arXiv:1404.4886}, 2014.

\bibitem{degond2013macroscopic}
P.~Degond, A.~Frouvelle, and J.~Liu.
\newblock Macroscopic limits and phase transition in a system of self-propelled
  particles.
\newblock {\em J. Nonlinear Sci.}, 23(3):427--456, 2013.

\bibitem{degond2015phase}
P.~Degond, A.~Frouvelle, and J.~Liu.
\newblock Phase transitions, hysteresis, and hyperbolicity for self-organized
  alignment dynamics.
\newblock {\em Arch. Ration. Mech. Anal.}, 216(1):63--115, 2015.

\bibitem{degond2012hydrodynamics}
P.~Degond and J.~Liu.
\newblock Hydrodynamics of self-alignment interactions with precession and
  derivation of the {L}andau--{L}ifschitz--{G}ilbert equation.
\newblock {\em Math. Models Methods Appl. Sci.}, 22(supp01):1140001, 2012.

\bibitem{degond2014evolution}
P.~Degond, J.~Liu, and C.~Ringhofer.
\newblock Evolution of wealth in a non-conservative economy driven by local
  nash equilibria.
\newblock {\em Phil. Trans. R. Soc. A}, 372(2028):20130394, 2014.

\bibitem{degond2015continuum}
P.~Degond, A.~Manhart, and H.~Yu.
\newblock A continuum model for nematic alignment of self-propelled particles.
\newblock {\em arXiv preprint arXiv:1509.03124}, 2015.

\bibitem{degond2008continuum}
P.~Degond and S.~Motsch.
\newblock Continuum limit of self-driven particles with orientation
  interaction.
\newblock {\em Math. Models Methods Appl. Sci.}, 18(supp01):1193--1215, 2008.

\bibitem{degond2015multi}
P.~Degond and L.~Navoret.
\newblock A multi-layer model for self-propelled disks interacting through
  alignment and volume exclusion.
\newblock {\em arXiv preprint arXiv:1502.05936}, 2015.

\bibitem{degond2007stochastic}
P.~Degond and C.~Ringhofer.
\newblock Stochastic dynamics of long supply chains with random breakdowns.
\newblock {\em SIAM J. Appl. Math.}, 68(1):59--79, 2007.

\bibitem{degond2015self}
P.~Degond and H.~Yu.
\newblock Self-organized hydrodynamics in an annular domain: Modal analysis and
  nonlinear effects.
\newblock {\em Math. Models Methods Appl. Sci.}, 25(03):495--519, 2015.

\bibitem{Figalli}
A.~Figalli, M.~Kang, and J.~Morales.
\newblock {Global well-posedness of the spatially homogeneous
  {K}olmogorov-{V}icsek model as a gradient flow}.
\newblock {\em ArXiv e-prints}, September 2015.

\bibitem{frouvelle2012continuum}
A.~Frouvelle.
\newblock A continuum model for alignment of self-propelled particles with
  anisotropy and density-dependent parameters.
\newblock {\em Math. Models Methods Appl. Sci.}, 22(07):1250011, 2012.

\bibitem{gallot1990riemannian}
S.~Gallot, D.~Hulin, and J.~Lafontaine.
\newblock {\em Riemannian geometry}, volume~3.
\newblock Springer, 1990.

\bibitem{Gamba}
I.~M. Gamba and M~Kang.
\newblock {Global weak solutions for {K}olmogorov-{V}icsek type equations with
  orientational interaction}.
\newblock {\em ArXiv e-prints}, February 2015.

\bibitem{gardiner}
C.~W. Gardiner.
\newblock {\em Stochastic methods}.
\newblock Springer-Verlag, Berlin--Heidelberg--New York--Tokyo, 1985.

\bibitem{golub2012matrix}
G.~H. Golub and C.~F. Van~Loan.
\newblock {\em Matrix computations}, volume~3.
\newblock JHU Press, 2012.

\bibitem{gregoire2004onset}
G.~Gr{\'e}goire and H.~Chat{\'e}.
\newblock Onset of collective and cohesive motion.
\newblock {\em Phys. Rev. Lett.}, 92(2):025702, 2004.

\bibitem{helbing2001traffic}
D.~Helbing.
\newblock Traffic and related self-driven many-particle systems.
\newblock {\em Rev. Modern Phys.}, 73(4):1067, 2001.

\bibitem{helbing2007dynamics}
D.~Helbing, A.~Johansson, and H.~Z. Al-Abideen.
\newblock Dynamics of crowd disasters: An empirical study.
\newblock {\em Phys. Rev. E}, 75(4):046109, 2007.

\bibitem{hsu2002stochastic}
E.~P. Hsu.
\newblock {\em Stochastic analysis on manifolds}, volume~38.
\newblock American Mathematical Soc., 2002.

\bibitem{huynh2009metrics}
D.~Q. Huynh.
\newblock Metrics for 3d rotations: Comparison and analysis.
\newblock {\em J. Math. Imaging Vision}, 35(2):155--164, 2009.

\bibitem{Hydro_limit}
N.~{Jiang}, L.~{Xiong}, and T.~{Zhang}.
\newblock {Hydrodynamic limits of the kinetic self-organized models}.
\newblock {\em ArXiv e-prints}, August 2015.

\bibitem{moakher2002means}
M.~Moakher.
\newblock Means and averaging in the group of rotations.
\newblock {\em SIAM J. Matrix Anal. Appl.}, 24(1):1--16, 2002.

\bibitem{parrish1997animal}
J.~K. Parrish and W.~M. Hamner.
\newblock {\em Animal groups in three dimensions: how species aggregate}.
\newblock Cambridge University Press, 1997.

\bibitem{sone2012kinetic}
Y.~Sone.
\newblock {\em Kinetic theory and fluid dynamics}.
\newblock Springer Science \& Business Media, 2012.

\bibitem{sznitman}
A.~Sznitman.
\newblock Topics in propagation of chaos.
\newblock In {\em Ecole d'Et{\'e} de Probabilit{\'e}s de Saint-Flour
  XIX—1989}, pages 165--251. Springer, 1991.

\bibitem{toner1995long}
J.~Toner and Y.~Tu.
\newblock Long-range order in a two-dimensional dynamical xy model: how birds
  fly together.
\newblock {\em Phys. Rev. Lett.}, 75(23):4326, 1995.

\bibitem{vicsek1995novel}
T.~Vicsek, A.~Czir{\'o}k, E.~Ben-Jacob, I.~Cohen, and O.~Shochet.
\newblock Novel type of phase transition in a system of self-driven particles.
\newblock {\em Phys. Rev. Lett.}, 75(6):1226, 1995.

\bibitem{vicsek2012collective}
T.~Vicsek and A.~Zafeiris.
\newblock Collective motion.
\newblock {\em Phys. Rep.}, 517(3):71--140, 2012.

\bibitem{zhang2010collective}
H.~Zhang, A.~Be’er, E.-L. Florin, and H.~L. Swinney.
\newblock Collective motion and density fluctuations in bacterial colonies.
\newblock {\em Proc. Natl. Acad. Sci.}, 107(31):13626--13630, 2010.

\end{thebibliography}
 
\end{document}